\newif\ifFull
\newif\ifCameraReady
\Fulltrue

\ifCameraReady
\documentclass[twoside,leqno,twocolumn]{article}
\usepackage{ltexpprt}
\fi
\ifFull
    \documentclass[11pt]{article}
    \usepackage[top=1in,left=1in,right=1in,bottom=1in]{geometry}
\fi
\usepackage{thm-restate}

\usepackage{xspace}
\usepackage{algorithm}
\usepackage[noend]{algpseudocode}
\usepackage{algorithmicx}
\usepackage{subcaption,graphicx,url,multirow,multicol,mathtools}
\usepackage{balance, afterpage}
\usepackage{paralist}
\usepackage{booktabs}
\usepackage{color}
\usepackage{arydshln}
\usepackage{enumitem}
\usepackage[most]{tcolorbox}
\usepackage{amssymb,stmaryrd}
\ifFull
\usepackage{amsthm}
\fi

\usepackage{bm}

\makeatletter
\newcommand{\removelatexerror}{\let\@latex@error\@gobble}
\makeatother

\algrenewtext{ParFor}[1]{\algorithmicparfor\ #1\ \algorithmicpardo}
\algrenewtext{EndParFor}{\algorithmicendparfor}
\algtext*{EndParFor}

\algdef{SE}[DOWHILE]{Do}{doWhile}{\algorithmicdo}[1]{\algorithmicwhile\ #1}

\definecolor{ao}{rgb}{0.0, 0.5, 0.0}

\ifCameraReady
\newtheorem{counter}{Theorem}[section]
\newtheorem{definition}[counter]{Definition}
\fi

\ifFull
\theoremstyle{plain}
\newtheorem{theorem}{Theorem}[section]
\newtheorem{definition}[theorem]{Definition}

\newtheorem{lemma}[theorem]{Lemma}
\newtheorem{corollary}[theorem]{Corollary}

\fi

\newcommand{\mysubsection}[1]{\smallskip\noindent {\bf #1.}}
\newcommand{\defn}[1]{\textbf{\emph{#1}}}
\usepackage[framemethod=tikz]{mdframed}

\newcommand{\op}{\omega_p}
\newcommand{\pmmdepth}[1]{\log{#1}}
\newcommand{\pmmspace}[1]{#1^{\op}}
\newcommand{\floor}[1]{\lfloor{#1\rfloor}}

\ifCameraReady
\usepackage[font=small,aboveskip=0pt,belowskip=0pt]{caption}
\usepackage[compact]{titlesec}
\titlespacing{\section}{3pt}{2pt}{2pt}
\titlespacing{\subsection}{2pt}{1pt}{1pt}
\titlespacing{\subsubsection}{2pt}{*0}{*0}

\titleformat{\section}
  {\normalfont\large\bfseries}{\thesection}{1em}{}

\usepackage{times,microtype}
\fi

\ifFull
    \usepackage{times}
\fi

\usepackage{hyperref}

\newcommand{\batch}{\Delta}

\newcommand{\high}{\mathcal{H}}
\newcommand{\low}{\mathcal{L}}
\newcommand{\hl}{\mathcal{HL}}
\newcommand{\hh}{\mathcal{HH}}
\newcommand{\lh}{\mathcal{LH}}
\newcommand{\lowl}{\mathcal{LL}}
\newcommand{\degarray}{\mathcal{D}}
\newcommand{\cT}{\mathcal{T}}
\newcommand{\id}{\mathsf{ID}}
\newcommand{\myparagraph}[1]{\smallskip\noindent {\bf #1.}}
\newcommand{\tridep}{\log^*(\batch + m)}

\newcommand{\scanpram}[0]{\emph{scan} \ensuremath{\mathsf{PRAM}}\renewcommand{\scanpram}[0]{scan \ensuremath{\mathsf{PRAM}}}}
\newcommand{\arbpram}[0]{\emph{arbitrary} CRCW \ensuremath{\mathsf{PRAM}}\renewcommand{\arbpram}[0]{arbitrary \ensuremath{\mathsf{PRAM}}}}
\newcommand{\pripram}[0]{\emph{priority} CRCW \ensuremath{\mathsf{PRAM}}\renewcommand{\pripram}[0]{priority CRCW \ensuremath{\mathsf{PRAM}}}}
\newcommand{\combpram}[0]{\emph{combining} CRCW \ensuremath{\mathsf{PRAM}}\renewcommand{\combpram}[0]{combining CRCW \ensuremath{\mathsf{PRAM}}}}

\newcommand{\multipram}[0]{\emph{multiprefix} CRCW \ensuremath{\mathsf{PRAM}}\renewcommand{\multipram}[0]{multiprefix CRCW \ensuremath{\mathsf{PRAM}}}}

\newcommand{\batchdynamic}[1]{batch-dynamic}

\definecolor{mygreen}{RGB}{20,140,80}
\definecolor{linkcolor}{RGB}{0,0,230}
\definecolor{mylightgray}{RGB}{230,230,230}
\definecolor{verylightgray}{RGB}{245,245,245}

\hypersetup{
     colorlinks=true,
     citecolor= mygreen,
     linkcolor= mygreen,
     urlcolor= mygreen
}

\newcounter{myalgctr}

\newtcolorbox{OuterBox}[1][]{

    breakable,
    enhanced,
    frame hidden,
    interior hidden,
    left=-5pt,
    right=-5pt,
    top=-5pt,
    float=p,
    boxsep=0pt,
    arc=0pt
#1}

\newtcolorbox{InnerBox}[1][]{

    enforce breakable,
    enhanced,
    colback=gray,
    colframe=white,
#1}

\newenvironment{tbox}{
\vspace{0.2cm}
\begin{tcolorbox}[width=\textwidth,
                  enhanced,
                  boxsep=2pt,
                  left=1pt,
                  right=1pt,
                  top=4pt,
                  boxrule=1pt,
                  arc=0pt,
                  colback=white,
                  colframe=black,
	              breakable
                  ]

}{
\end{tcolorbox}
}

\newcommand{\tboxhrule}[0]{\vspace{0.1cm} {\color{black} \hrule} \vspace{0.2cm}}

\newenvironment{titledtbox}[1]{\begin{tbox}#1 \tboxhrule}{\end{tbox}}

\newenvironment{mdframedalg}[1]{\begin{algorithm}[htpb]\caption{#1}}{\end{algorithm}}

\makeatletter
\newenvironment{breakablealgorithm}
  {

   \begin{center}
     \refstepcounter{algorithm}

     \hrule height.8pt depth0pt \kern2pt

     \renewcommand{\caption}[2][\relax]{

       {\raggedright\textbf{\fname@algorithm~\thealgorithm} ##2\par}

       \ifx\relax##1\relax 

         \addcontentsline{loa}{algorithm}{\protect\numberline{\thealgorithm}##2}

       \else 

         \addcontentsline{loa}{algorithm}{\protect\numberline{\thealgorithm}##1}

       \fi
       \kern2pt\hrule\kern2pt
     }
  }{

     \kern2pt\hrule\relax

   \end{center}
  }
\makeatother

\newcommand{\algorithmfootnote}[2][\footnotesize]{

  \let\old@algocf@finish\@algocf@finish

  \def\@algocf@finish{\old@algocf@finish

    \leavevmode\rlap{\begin{minipage}{\linewidth}
    #1#2
    \end{minipage}}

  }

}

\begin{document}
\sloppy
\ifFull
\pagestyle{plain}
\fi
\date{}

\title{Parallel Batch-Dynamic $k$-Clique Counting}

\author{Laxman Dhulipala \\ MIT CSAIL \\ laxman@mit.edu \and
    Quanquan C. Liu \\ MIT CSAIL \\ quanquan@mit.edu \and
    Julian Shun \\ MIT CSAIL \\ jshun@mit.edu \and
    Shangdi Yu \\  MIT CSAIL \\ shangdiy@mit.edu}

\maketitle

\begin{abstract}
In this paper, we study new batch-dynamic algorithms for the
$k$-clique counting problem, which are dynamic algorithms where the
updates are \emph{batches} of edge insertions and deletions. We study
this problem in the parallel setting, where the goal is to obtain
algorithms with low (polylogarithmic) depth. Our first result is a new
parallel batch-dynamic triangle counting algorithm with
$O(\Delta\sqrt{\Delta+m})$ amortized work and $O(\log^* (\Delta+m))$
depth with high probability, and $O(\Delta+m)$ space for a batch of
$\Delta$ edge insertions or deletions.  Our second result is an
algebraic algorithm based on parallel fast matrix multiplication.
Assuming that a parallel fast matrix multiplication algorithm exists
with parallel matrix multiplication constant $\omega_p$, the same
algorithm solves dynamic $k$-clique counting with
$O\left(\min\left(\batch m^{\frac{(2k - 1)\op}{3(\op + 1)}},
(\batch+m)^{\frac{2(k + 1)\op}{3(\op + 1)}}\right)\right)$ amortized
work and $O(\log (\Delta+m))$ depth with high probability, and
$O\left((\batch+m)^{\frac{2(k + 1)\op}{3(\op + 1)}}\right)$ space.
Using a recently developed parallel $k$-clique counting algorithm, we
also obtain a simple batch-dynamic algorithm for $k$-clique counting
on graphs with arboricity $\alpha$ running in
$O(\Delta(m+\Delta)\alpha^{k-4})$ expected work and $O(\log^{k-2} n)$
depth with high probability, and $O(m + \Delta)$ space.  Finally, we
present a multicore CPU implementation of our parallel batch-dynamic
triangle counting algorithm. On a 72-core machine with two-way
hyper-threading, our implementation achieves 36.54--74.73x parallel speedup,
and in certain cases achieves significant speedups over existing parallel algorithms for the problem, which are not theoretically-efficient.

\end{abstract}

\section{Introduction}

Subgraph counting algorithms are fundamental graph analysis tools,
with numerous applications in network classification in domains
including social network analysis and bioinformatics. A particularly
important type of subgraph for these applications is the triangle, or
$3$-clique---three vertices that are all mutually
connected~\cite{newman2003structure}. Counting the number of
triangles is a basic and fundamental task that is used in numerous
social and network science measurements~\cite{granovetter1977strength,
watts1998collective}.

In this paper, we study the triangle counting problem and its
generalization to higher cliques from the perspective of dynamic
algorithms. A $k$-clique consists of $k$ vertices and all $k
\choose{2}$ possible edges among them (for applications of
$k$-cliques, see, e.g.,~\cite{hanneman05-introduction}). As many
real-world graphs change rapidly in real-time, it is crucial to design
dynamic algorithms that efficiently maintain $k$-cliques upon updates,
since the cost of re-computation from scratch can be prohibitive.
Furthermore, due to the fact that dynamic updates can occur at a rapid
rate in practice, it is increasingly important to design
\defn{\batchdynamic{}} algorithms which can take arbitrarily large
batches of updates (edge insertions or deletions) as their input.
Finally, since the batches, and corresponding update complexity can be
large, it is also desirable to use parallelism to speed-up maintenance
and design algorithms that map to modern parallel architectures.

Due to the broad applicability of $k$-clique counting in practice and
the fact that $k$-clique counting is a fundamental theoretical problem
of its own right, there has been a large body of prior work on the
problem. Theoretically, the fastest static algorithm for arbitrary
graphs uses fast matrix multiplication, and counts $3\ell$ cliques in
$O(n^{\ell\omega})$ time where $\omega$ is the matrix multiplication
exponent~\cite{nevsetvril1985complexity}. Considerable effort has also
been devoted to efficient combinatorial algorithms.  Chiba and
Nishizeki~\cite{Chiba1985} show how to compute $k$-cliques in
$O(\alpha^{k-2}m)$ work, where $m$ is the number of edges in the graph
and $\alpha$ is the arboricity of the graph. This algorithm was
recently parallelized by Danisch et al.~\cite{Danisch2018} (although
not in polylogarithmic depth). Worst-case optimal join algorithms can
perform $k$-clique counting in $O(m^{k/2})$ work as a special
case~\cite{Ngo2018,AbergerLTNOR17}. Alon, Yuster, and Zwick~\cite{AYZ97}
design an algorithm for triangle counting in the sequential model,
based on fast matrix multiplication.
Eisenbrand and Grandoni~\cite{EG04} then extend this result to
$k$-clique counting based on fast matrix multiplication.
Vassilevska designs
a space-efficient combinatorial algorithm for $k$-clique
counting~\cite{vassilevska2009efficient}.  Finocchi et al.\  give
clique counting algorithms for MapReduce~\cite{Finocchi2015}.  Jain
and Seshadri provide probabilistic algorithms for estimating clique
counts~\cite{Jain2017}. The $k$-clique problem is also a classical
problem in parameterized-complexity, and is known to be
$W[1]$-complete~\cite{downey1995fixed}.

The problem of maintaining $k$-cliques under dynamic updates began
more recently. Eppstein et al.~\cite{Eppstein2009,Eppstein2012} design
sequential dynamic algorithms for maintaining size-3 subgraphs in
$O(h)$ amortized time and $O(mh)$ space and size-4 subgraphs in
$O(h^2)$ amortized time and $O(mh^2)$ space, where $h$ is the
$h$-index of the graph ($h=O(\sqrt{m})$). Ammar et al. extend the
worst-case optimal join algorithms to the parallel and dynamic
setting~\cite{Ammar2018}. However, their update time is not better
than the static worst-case optimal join algorithm.  Recently, Kara et
al.~\cite{KNNOZ19} present a sequential dynamic algorithm for
maintaining triangles in $O(\sqrt{m})$ amortized time and $O(m)$
space. Dvorak and Tuma~\cite{Dvorak2013} present a dynamic algorithm
that maintains $k$-cliques as a special case in $O(\alpha^{k-2} \log
n)$ amortized time and $O(\alpha^{k-2} m)$ space by using
low out-degree orientations for graphs with arboricity $\alpha$.

\myparagraph{Designing Parallel Batch-Dynamic Algorithms}
Traditional dynamic algorithms receive and apply updates one at a
time. However, in the \defn{parallel \batchdynamic{}} setting, the
algorithm receives \emph{batches of updates} one after the other,
where each batch contains a mix of edge insertions and deletions.
Unlike traditional dynamic algorithms, a parallel batch-dynamic
algorithm can apply \emph{all} of the updates together, and also take
advantage of parallelism while processing the batch.
We note that the edges inside of a batch may also be ordered (e.g., by
a timestamp). If there are duplicate edge insertions within a batch,
or an insertion of an edge followed by its deletion, a batch-dynamic
algorithm can easily remove such redundant or nullifying updates.

The key challenge is to design the algorithm so that updates can be
processed in parallel while ensuring low work and depth bounds.
The only existing parallel batch-dynamic algorithms for $k$-clique counting are triangle counting algorithms by
Ediger et al.~\cite{Ediger2010} and Makkar et
al.~\cite{Makkar2017}, which take linear work per update in the worst case.
The
algorithms in this paper make use of efficient data structures such as
parallel hash tables, which let us perform parallel batches of edge insertions
and deletions  with better work and (polylogarithmic) depth bounds.
To the best of our knowledge, no prior work has designed dynamic algorithms for the problem that support parallel batch updates with non-trivial theoretical guarantees.

Theoretically-efficient parallel dynamic (and batch-dynamic)
algorithms have been designed for a variety of other graph problems,
including minimum spanning
tree~\cite{Kopelowitz2018,Ferragina1994,Das1994}, Euler tour
trees~\cite{TsengDB19},
connectivity~\cite{Simsiri2018,Acar2019,Ferragina1994}, tree
contraction~\cite{Reif1994,Acar2017}, and depth-first
search~\cite{Khan2017}. Very recently, parallel dynamic algorithms
were also designed for the Massively Parallel Computation (MPC)
setting~\cite{italiano2019dynamic, dhulipala2020parallel}.

\ifFull
\myparagraph{Other Related Work}
There has been significant amount of
work on practical parallel algorithms for the case of static 3-clique
counting, also known as triangle counting.
(e.g.,~\cite{Suri2011,Arifuzzaman2013,Park2013,Park14, ShunT2015},
among many others).
Due to the importance of the problem, there is
even an annual competition for parallel triangle counting
solutions~\cite{GraphChallenge}.
Practical static counting algorithms for the special cases of
$k=4$ and $k=5$ have also been
developed~\cite{Hocevar2014,Elenberg2016,Pinar2017,AhmedNRDW17,Dave2017}.

Dynamic algorithms have been studied in
distributed models of computation under the framework of
\emph{self-stabilization}~\cite{schneider1993self}. In this setting, the system
undergoes various changes, for example topology changes, and must
quickly converge to a stable state. Most of the existing work in this
setting focuses on a single change per round~\cite{censor2016optimal,
bonne19distributedclique, assadi2019fully}, although
algorithms studying multiple changes per round have been considered
very recently~\cite{bamberger2019local, censor2019fast}.
Understanding
how these algorithms relate to parallel \batchdynamic{} algorithms is
an interesting question for future work.
\fi

\myparagraph{Summary of Our Contributions}
In this paper, we design parallel algorithms in the \batchdynamic{}
setting, where the algorithm receives a batch of $\Delta \geq 1$ edge
updates that can be processed in parallel.
Our focus is on parallel \batchdynamic{}
algorithms that admit strong theoretical bounds on their work and
have polylogarithmic depth with high probability. Note that although
our work bounds may be amortized, our depth will be polylogarithmic
with high probability, leading to efficient $\mathsf{RNC}$ algorithms.
As a special case of our results, we obtain algorithms for
parallelizing single updates ($\Delta=1$). We first design a parallel
\batchdynamic{} triangle counting algorithm based on the sequential
algorithm of Kara et al.~\cite{KNNOZ19}. For triangle counting, we
obtain an algorithm that takes $O(\Delta\sqrt{\Delta+m})$ amortized
work and $O(\log^* (\Delta+m))$ depth w.h.p.\footnote{We use ``with high probability''
  (w.h.p.) to mean with probability at least $1-1/n^c$ for any constant
  $c>0$.}
assuming a fetch-and-add instruction that runs in $O(1)$ work and
depth, and runs in $O(\Delta+m)$ space.
The work of our parallel algorithm matches that of the sequential algorithm of performing one update at a time (i.e., it is work-efficient), and we can perform all updates in parallel with low depth.

We then present a new parallel \batchdynamic{} algorithm based on
fast matrix multiplication. Using the best currently known parallel
matrix multiplication~\cite{Williams12,LeGall14}, our algorithm
dynamically maintains the number of $k$-cliques in
$O\left(\min\left(\batch m^{0.469k - 0.235}, (\batch+m)^{0.469k +
0.469}\right)\right)$ amortized work w.h.p.\ per batch of $\batch$ updates
where $m$ is defined
as the maximum number of edges in the graph before and after
all updates in the batch are applied. Our approach
is based on the algorithm of~\cite{AYZ97,EG04,nevsetvril1985complexity},
and maintains triples of
$k/3$-cliques that together form $k$-cliques. The depth is $O(\log
(\Delta+m))$ w.h.p.\ and the space is $O\left((\batch+m)^{0.469k +
0.469}\right)$. Our results also imply an amortized time bound of
$O\left(m^{0.469k - 0.235}\right)$ per update for dense graphs in the
sequential setting.
Of potential independent interest, we present the first proof of logarithmic
depth in the parallelization of any tensor-based fast matrix multiplication
algorithms.
We also give a simple
batch-dynamic $k$-clique listing algorithm, based on enumerating
smaller cliques and intersecting them with edges in the batch. The
algorithm runs in $O(\Delta(m+\Delta)\alpha^{k-4})$ expected work,
$O(\log^{k-2}n)$ depth w.h.p., and $O(m + \batch)$ space.

Finally, we implement our new parallel batch-dynamic triangle counting
algorithm for multicore CPUs, and present some experimental results on large graphs and
with varying batch sizes using a 72-core machine with two-way hyper-threading.
We found our parallel implementation to be much faster than the
multicore implementation of Ediger et al.~\cite{Ediger2010}. We also
developed an optimized multicore implementation of the GPU algorithm
by Makkar et al.~\cite{Makkar2017}. We found that our new algorithm is
up to an order of magnitude faster than our CPU implementation of the
Makkar et al.\ algorithm, and our new algorithm achieves 36.54--74.73x
parallel speedup on 72 cores with hyper-threading.  Our code is
publicly available at {\small \url{https://github.com/ParAlg/gbbs}}.

\section{Preliminaries}\label{sec:prelims}
Given an undirected graph $G = (V, E)$ with $n$ vertices and $m$
edges, and an integer $k$, a \defn{$k$-clique} is defined as a set
of $k$ vertices $v_1,\ldots,v_k$ such that for all $i\neq j$, $(v_i,
v_j) \in E$. The \defn{$k$-clique count} is the total number of
$k$-cliques in the graph. The \defn{dynamic $k$-clique problem}
maintains the number of $k$-cliques in the graph upon edge insertions
and deletions, given individually or in a batch.  The
\defn{arboricity} $\alpha$ of a graph is the minimum number of
forests that the edges can be partitioned into and its value is
between $\Omega(1)$ and $O(\sqrt{m})$~\cite{Chiba1985}.

In this paper, we analyze algorithms in the work-depth model, where
the \defn{work} of an algorithm is defined to be the total number of
operations done, and the \defn{depth} is defined to be the longest
sequential dependence in the computation (or the computation time
given an infinite number of processors)~\cite{JaJa92}.
Our algorithms
can run in the PRAM model or the fork-join model with arbitrary forking.
We use the
concurrent-read concurrent-write (CRCW) model, where reads and writes
to a memory location can happen concurrently.  We assume either that
concurrent writes are resolved arbitrarily, or are reduced together
(i.e., fetch-and-add PRAM).

We use the following primitives throughout the paper.
\defn{Approximate compaction} takes a set of $m$ objects in the
range $[1, n]$ and allocates them unique IDs in the range $[1, O(m)]$.
The primitive is useful for filtering (i.e., removing) out a set of
obsolete elements from an array of size $n$, and mapping the remaining
$m$ elements to a sparse array of size $O(m)$. Approximate
compaction can be implemented in $O(n)$ work and $O(\log^* n)$
depth w.h.p.~\cite{Gil91a}.  We also use a \defn{parallel hash table} which
supports $n$ operations (insertions, deletions) in $O(n)$ work and
$O(\log^* n)$ depth w.h.p., and $n$ lookup operations
in $O(n)$ work and $O(1)$ depth~\cite{Gil91a}.

Our algorithms in this paper make use of the widely used
\defn{atomic-add} instruction. An atomic-add instruction takes a
memory location and atomically increments the value stored at the
location. In this paper, we assume that the atomic-add instruction can
be implemented in $O(1)$ work and depth.  Our algorithms can also be
implemented in a model without atomic-add in the same work,
a multiplicative $O(\log n)$ factor increase in the depth, and space
proportional to the number of atomic-adds done in parallel.

\ifFull
\section{Technical Overview}

In this section, we present a high-level technical overview of our
approach in this paper.

\ifCameraReady
\subsection{Parallel Batch-Dynamic Triangle Counting}~
\fi
\ifFull
    \subsection{Parallel Batch-Dynamic Triangle Counting}
\fi

Our parallel \batchdynamic{} triangle counting algorithm is based on a
recently proposed sequential dynamic algorithm due to Kara et
al.~\cite{KNNOZ19}. They describe their algorithm in the database
setting, in the context of dynamically maintaining the result of a
database join. We provide a self-contained description of their
sequential algorithm in Appendix~\ref{app:triangle}.

\myparagraph{High-Level Approach} The basic idea of the algorithm
from~\cite{KNNOZ19} is to partition the vertex set using degree-based
thresholding. Roughly, they specify a threshold $t=\Theta(\sqrt{m})$,
and classify all vertices with degree less than $t$ to be low-degree,
and all vertices with degree larger than $t$ to be high-degree. This
thresholding technique is widely used in the design of fast static
triangle counting and $k$-clique counting algorithms,
(e.g.,~\cite{nevsetvril1985complexity, AYZ97}). Observe that if we
insert an edge $(u,v)$ incident to a low-degree vertex, $u$, we can
enumerate all vertices $w$ in $N(u)$ in $O(\sqrt{m})$ expected time and check
if $(u,v,w)$ forms a triangle (checking if the $(v,w)$ edge is present
in $G$ can be done by storing all edges in a hash table). In this way,
edge updates incident to low-degree vertices are handled relatively
simply. The more interesting case is how to handle edge updates
between high-degree vertices. The main problem is that a single edge
insertion $(u,v)$ between two high-degree vertices can cause up to
$O(n)$ triangles to appear in $G$, and enumerating all of these would
require $O(n)$ work---potentially much more than
$O(\sqrt{m})$. Therefore, the algorithm maintains an auxiliary data
structure, $\cT$, over wedges ($2$-paths). $\cT$ stores for every pair
of high-degree vertices $(v,w)$, the number of low-degree vertices $u$
that are connected to both $v$ and $w$ (i.e., $(u,v)$ and $(u,w)$ are
both in $E$). Given this structure, the number of triangles formed by
the insertion of the edge $(v,w)$ going between two high-degree
vertices can be found in $O(1)$ time by checking the count for $(v,w)$
in $\cT$. Updates to $\cT$ can be handled in $O(\sqrt{m})$ time, since
$\cT$ need only be updated when a low-degree vertex inserts/deletes a
neighbor, and the number of entries in $\cT$ that are affected is at
most $t$. Some additional care needs to be taken when
specifying the threshold $t$ to handle re-classifying vertices (going
from low-degree to high-degree, or vice versa), and also to handle
rebuilding the data structures, which leads to a bound of
$O(\sqrt{m})$ amortized work per update for the algorithm.

\myparagraph{Incorporating Batching and Parallelism} The input to the
parallel \batchdynamic{} algorithm is a batch containing (possibly) a
mix of edge insertions and deletions (vertex insertions and deletions
can be handled by inserting or deleting its incident edges). For
simplicity, and without any loss in our asymptotic bounds, our
algorithm handles insertions and deletions separately. The algorithm
first removes all \emph{nullifying} updates, which are updates that
have no effect after applying the entire batch (i.e., an insertion which
is subsequently deleted within the same batch, an insertion of an edge
that already exists or a deletion of an edge that doesn't exist). This
can easily be done within the bounds using basic parallel primitives.
The algorithm then updates tables representing the adjacency
information of both low-degree and high-degree vertices in parallel.
To obtain strong parallel bounds, we represent these sets using
parallel hash tables. For each insertion (deletion), we then determine
the number of new triangles that are created (deleted). Since a given triangle
could incorporate multiple edges within the same batch of insertions
(deletions), our algorithm must carefully ensure that the triangle is
counted only once, assigning each new inserted (deleted) triangle
uniquely to one of the updates forming it. We then update the overall
triangle count with the number of distinct triangles inserted
(deleted) into the graph by the current batch of insertions
(deletions). The remaining work of the algorithm cleans up mutable
state in the
hash tables, and also migrates vertices between low-degree and
high-degree states.

\myparagraph{Worst-Case Optimality}
Our work bounds match the combinatorial lower
bound obtained via a fine-grained reduction from triangle detection which is
conjectured to take $m^{3/2 - o(1)}$ work (by the
\emph{Strong Triangle conjecture} of~\cite{AW14} for combinatorial algorithms). The combinatorial
lower bound for the Strong Triangle conjecture is based on the standard
lower bound conjecture for combinatorial algorithms that solve Boolean Matrix Multiplication
(BMM).
Our reduction proceeds as follows.
Given any input graph to the triangle detection problem, we divide the edges
into batches of edge insertions arbitrarily without knowledge of the existence
of (any) triangles.
Then, the batches of updates are
applied one after the other. Suppose the amortized work per update
for this procedure is $O(X)$. Then, the total work for applying all the batches of
updates is $O(Xm)$. The algorithm returns the count of the number of triangles
in the graph after applying all batches of updates.
In this case, the algorithm when run over all the batches
solves the static problem of triangle detection in the original input graph. If
the number of triangles counted by the algorithm after the last batch is $0$,
then there does not exist a triangle in the original input graph; otherwise,
there exists a triangle in the original input graph.
If $X = m^{1/2 - \Omega(1)}$, then we violate the Strong Triangle conjecture.
Thus, our work bound is conditionally optimal up to sub-polynomial factors
by the Strong Triangle conjecture.

It is an interesting open question to consider whether one can
obtain $O(1)$ depth bounds
on the CRCW PRAM.

\ifCameraReady
\subsection{Dynamic $k$-Clique Counting via Fast Static Parallel
  Algorithms}~
\fi
\ifFull
\subsection{Dynamic $k$-Clique Counting via Fast Static Parallel
  Algorithms}
\fi

Next, we present a very simple, and potentially practical
algorithm for dynamically maintaining the number of $k$-cliques based
on statically enumerating smaller cliques in the graph, and
intersecting the enumerated cliques with the edge updates in the input
batch. The algorithm is space-efficient, and is asymptotically more
efficient than other methods for sparse graphs. Our algorithm is based
on a recent and concurrent work proposing a work-efficient parallel
algorithm for counting $k$-cliques in  $O(m\alpha^{k-2})$ expected work and
polylogarithmic depth w.h.p.~\cite{shi2020parallel}. Using this algorithm, we
show that updating the $k$-clique count for a batch of $\Delta$
updates can be done in $O(\Delta(m+\Delta)\alpha^{k-4})$ expected work, and
  $O(\log^{k-2} n)$ depth w.h.p., using $O(m + \batch)$
  space.
We do this by using the static
algorithm to (i) enumerate all $(k-2)$-cliques, and (ii) checking
whether each $(k-2)$-clique forms a $k$-clique with an edge in the
batch.

\ifCameraReady
\subsection{Dynamic $k$-Clique via Fast Matrix Multiplication}~
\fi
\ifFull
\subsection{Dynamic $k$-Clique via Fast Matrix Multiplication}
\fi

We then present a parallel batch-dynamic $k$-clique counting algorithm
using parallel fast matrix multiplication (MM).
Our algorithm is inspired
by the static triangle counting algorithm of Alon, Yuster, and Zwick
(AYZ)~\cite{AYZ97} and the static $k$-clique counting algorithm
of~\cite{EG04} that uses MM-based triangle counting.  We present a new
batch-dynamic algorithm that obtains better bounds than the simple algorithm
based on static smaller-clique enumeration above (and also presented in
Section~\ref{sec:arboricityclique}) for $k > 9$. To the best of our
knowledge, this is also the best bound for
dynamic triangle counting on dense graphs in the sequential model.
Specifically, assuming a parallel matrix multiplication exponent of
$\omega_p$, our algorithm handles batches of $\Delta$ edge
insertions/deletions using $O\left(\min\left(\Delta m^{\frac{(2k -
    3)\op}{3(1+\op)}}, (m +
\batch)^{\frac{2k\op}{3(1+\op)}}\right)\right)$ work and $O(\log m)$
depth w.h.p., in $O\left((m + \batch)^{\frac{2k\op}{3(1+\op)}}\right)$
space, where $m$ is the number of edges in the graph before
applying the batch of updates.
To the best of our knowledge, the sequential (batch-dynamic)
version of our algorithm also provides the best bounds for dynamic
triangle counting in the sequential model for dense graphs for such
values of $k$ (assuming that we use the best currently known matrix
multiplication algorithm)~\cite{Dvorak2013}.

\myparagraph{High-Level Approach and Techniques} For a given graph $G
= (V, E)$, we create an auxiliary graph $G' = (V', E')$ with vertices
and edges representing cliques of various sizes in $G$. For a given
$k$-clique problem, vertices in $V'$ represent cliques of size $k/3$
in $G$ and edges $(u, v)$ between vertices $u, v \in V'$ represent
cliques of size $2k/3$ in $G$. Thus, a triangle in $G'$ represents a
$k$-clique in $G$. Specifically, there exists exactly ${k \choose
  k/3}{2k/3 \choose k/3}$ different triangles in $G'$ for each clique
in $G$.

Given a batch of edge insertions and deletions to $G$, we create a set
of edge insertions and deletions to $G'$. An edge is inserted in $G'$
when a new $2k/3$-clique is created in $G$ and an edge is deleted in
$G'$ when a $2k/3$-clique is destroyed in $G$.  Suppose, for now,
that we have a dynamic algorithm for processing the edge
insertions/deletions into $G'$.  Counting the number of triangles in
$G'$ after processing all edge insertions/deletions and
dividing by ${k \choose k/3}{2k/3 \choose k}$ provides us with the
exact number of cliques in $G$.

There are a number of challenges that we must deal with when
formulating our dynamic triangle counting algorithm for counting the
triangles in $G'$:
\begin{enumerate}
\item We cannot simply count all the triangles in $G'$ after
  inserting/deleting the new edges as this does not perform better
  than a trivial static algorithm.
\item Any trivial dynamization of the AYZ algorithm will not be able
  to detect all new triangles in $G'$. Specifically, because the AYZ
  algorithm counts all triangles containing a low-degree vertex
  separately from all triangles containing only high-degree vertices,
  if an edge update only occurs between high-degree vertices, a
  trivial dynamization of the algorithm will not be able to detect any
  triangle that the two high-degree endpoints make with low-degree
  vertices.
\end{enumerate}

To solve the first challenge, we dynamically count \emph{low-degree}
and \emph{high-degree} vertices in different ways. Let $\ell=k/3$ and
$M = 2m + 1$. For some value of $0<t<1$, we define \emph{low-degree}
vertices to be vertices that have degree less than $M^{t\ell}/2$ and
\emph{high-degree} vertices to have degree greater than
$3M^{t\ell}/2$.  Vertices with degrees in the range $[M^{t\ell}/2,
  3M^{t\ell}/2]$ can be classified as either low-degree or
high-degree.  We determine the specific value for $t$ in
Lemma~\ref{lem:t-value}. We perform rebalancing of the data structures
as needed as they handle more updates. For low-degree vertices, we
only count the triangles that include at least one newly
inserted/deleted edge, at least one of whose endpoints is low-degree.
This means that we do not need to count any pre-existing triangles
that contain at least one low-degree vertex.  For the high-degree
vertices, because there is an upper bound on the maximum number of
such vertices in the graph, we update an adjacency matrix $A$
containing only edges between high-degree vertices.  At the end of all
of the edge updates, computing $A^3$ gives us a count of all of the
triangles that contain three high-degree vertices.

This procedure immediately then leads to our second challenge. To
solve this second challenge, we make the observation (proven in
Lemma~\ref{lem:one-low-high}) that if there exists an edge update
between two high-degree vertices that creates or destroys a triangle
that contains a low-degree vertex in $G'$, then there \emph{must}
exist at least one new edge insertion/deletion \emph{that creates or
  destroys a triangle representing the same clique} to that low-degree
vertex in the same batch of updates to $G'$. Thus, we can use one of
those edge insertions/deletions to determine the new clique that was
created and, through this method, find all triangles containing at
least one low-degree vertex and at least one new edge update. Some
care must be observed in implementing this procedure in order to not
increase the runtime or space usage; such details can be found in
Section~\ref{sec:alg-overview}.

\myparagraph{Incorporating Batching and Parallelism} When
dealing with a batch of updates containing both edge insertions and
deletions, we must be careful when vertices switch from being
high-degree to being low-degree, and vice versa. If we intersperse the
edge insertions with the edge deletions, there is the possibility that
a vertex switches between low and high-degree multiple times in a
single batch.  Thus, we batch all edge deletions together and perform
these updates first before handling the edge insertions.  After
processing the batch of edge deletions, we must subsequently move any
high-degree vertices that become low-degree to their correct data
structures. After dealing with the edge insertions, we must similarly
move any low-degree vertices that become high-degree to the correct
data structures.  Finally, for triangles that contain more than one
edge update, we must account for potential double counting by
different updates happening in parallel.  Such challenges are
described and dealt with in Section~\ref{sec:alg-overview} and
Algorithm~\ref{alg:mmclique}.

\subsection{Implementation and Experimental Evaluation}
We present an optimized implementation of our new parallel
batch-dynamic triangle counting algorithm using parallel primitives
from the Graph Based Benchmark Suite
(GBBS)~\cite{dhulipala2018theoretically}, and concurrent hash tables~\cite{shun2014phase}
to represent our data structures.  We ran experiments on varying batch
sizes for both insertions and deletions for several large graphs (the
Orkut and Twitter graphs, as well as rMAT graphs of varying densities)
using a 72-core machine with two-way hyper-threading, and obtained
parallel speedups of between 36.54--74.73x.  We also compared our
performance to the algorithms by Ediger et al.~\cite{Ediger2010} and
Makkar et al.~\cite{Makkar2017} (we note that Makkar et al.\ provide a
GPU implementation, and we implemented a multicore CPU version of
their algorithm), which take linear work per update in the worst
case. We found that our Makkar et al.\ implementation outperformed the
multicore implementation by Ediger et al.  Furthermore, our new
algorithm achieves significant speedups (up to an order of magnitude)
over the Makkar et al.\ implementation on graphs with high-degree
vertices (the Twitter graph and dense rMAT graphs), as well as on
smaller batch sizes. In contrast, the Makkar et al.\ implementation
outperforms our new algorithm for the smaller Orkut graph, which does
not contain vertices with very high degree. These results are
consistent with the theoretical bounds of the algorithms---the work
per update of our algorithm is $O(\sqrt{m})$, whereas the work per
update of the Makkar et al.\ algorithm is linear in the degrees of the
affected vertices.

\fi

\newcommand{\ins}{\mathtt{insert}}
\newcommand{\del}{\mathtt{delete}}
\newcommand{\ctriang}[3]{\mathtt{count\_triangles(#1, #2, #3)}}
\newcommand{\minreb}[1]{\mathtt{minor\_rebalance(#1)}}
\newcommand{\remove}[1]{\mathtt{remove\_useless\_updates}(#1)}
\newcommand{\batchset}{\mathcal{B}}
\newcommand{\markins}[1]{\mathtt{mark\_inserted\_edges}(#1)}
\newcommand{\markdels}[1]{\mathtt{mark\_deleted\_edges}(#1)}
\newcommand{\tins}[1]{\mathtt{update\_table\_insertions}(#1)}
\newcommand{\tdel}[1]{\mathtt{update\_table\_deletions}(#1)}

\newcommand{\tinsnoarg}{\mathtt{update\_table\_insertions}}
\newcommand{\tdelnoarg}{\mathtt{update\_table\_deletions}}
\newcommand{\tup}[1]{t^{(u, v)}_{#1}}

\algblock{ParFor}{EndParFor}
\algnewcommand\algorithmicparfor{\textbf{parfor}}
\algnewcommand\algorithmicpardo{\textbf{do}}
\algnewcommand\algorithmicendparfor{}
\algrenewtext{ParFor}[1]{\algorithmicparfor\ #1\ \algorithmicpardo}
\algrenewtext{EndParFor}{\algorithmicendparfor}
\algtext*{EndParFor}{}

\section{Parallel Batch-Dynamic Triangle Counting}\label{sec:batch-triangle}
\ifFull
    We now present our
\fi
parallel \batchdynamic{} triangle counting
algorithm, which is based on the $O(m)$ space and $O(\sqrt{m})$
amortized update, sequential, dynamic algorithm of Kara et
al.~\cite{KNNOZ19}. Theorem~\ref{thm:linear-space-update}
summarizes the guarantees of our algorithm.

\begin{theorem}\label{thm:linear-space-update}
  There exists a parallel \batchdynamic{} triangle
  counting algorithm that requires $O(\Delta(\sqrt{\Delta+m}))$ amortized
   work and $O(\tridep)$ depth with high probability, and
   $O(\Delta+m)$ space for a batch of $\Delta$ edge updates.

\end{theorem}

Our algorithm is work-efficient and achieves a significantly lower depth for a batch of updates than applying the updates one at a time using the sequential algorithm of~\cite{KNNOZ19}.
We provide a detailed description of the fully dynamic sequential
algorithm of~\cite{KNNOZ19}
\ifFull
in Appendix~\ref{app:triangle}
\fi
for
reference,\footnote{Kara et al.~\cite{KNNOZ19} described their algorithm for counting directed 3-cycles in relational databases, where each triangle edge is drawn from a different relation, and we simplified it for the case of undirected graphs.} and a brief high-level overview of that algorithm in
this section.

\ifFull
\subsection{Sequential Algorithm Overview}\label{sec:seq}
\fi

Given a graph $G= (V, E)$ with $n = |V|$ vertices and $m = |E|$ edges,
let $M = 2m+1$, $t_1 = \sqrt{M}/2$, and $t_2 = 3\sqrt{M}/2$. We classify
a vertex as \defn{low-degree} if its degree is at most $t_1$ and
\defn{high-degree} if its degree is at least $t_2$. Vertices with
degree in between $t_1$ and $t_2$ can be classified either way.

\myparagraph{Data Structures} The algorithm partitions the edges into
four edge-stores $\hh$, $\hl$, $\lh$, and $\lowl$ based on a
degree-based partitioning of the vertices.  $\hh$ stores all of the
edges $(u, v)$, where both $u$ and $v$ are high-degree. $\hl$ stores
edges $(u,v)$, where $u$ is high-degree and $v$ is low-degree. $\lh$
stores the edges $(u, v)$, where $u$ is low-degree and $v$ is
high-degree. Finally, $\lowl$ stores edges $(u,v)$, where both $u$ and $v$ are
low-degree.

The algorithm also maintains a wedge-store $\cT$ (a wedge is a triple
of distinct vertices $(x,y,z)$ where both $(x, y)$ and $(y, z)$ are edges
in $E$). For each pair of high-degree vertices $u$ and $v$, the
wedge-store $\cT$ stores the number of wedges $(u, w, v)$, where $w$
is a low-degree vertex. $\cT$ has the property that given an edge
insertion (resp.\ deletion) $(u, v)$ where both $u$ and $v$ are
high-degree vertices, it returns the number of wedges $(u, w, v)$,
where $w$ is low-degree, that $u$ and $v$ are part of in $O(1)$
expected time. $\cT$ is implemented via a hash table indexed by pairs
of high-degree vertices that stores the number of wedges for each
pair.

Finally, we have an array containing the degrees of each vertex,
$\degarray$.

\myparagraph{Initialization}
Given a graph with $m$ edges, the algorithm first initializes the
triangle count $C$ using a static triangle counting algorithm in
$O(\alpha m)=O(m^{3/2})$ work and $O(m)$ space~\cite{Latapy2008}. The $\hh$,
$\hl$, $\lh$, and $\lowl$ tables are created by scanning all edges
in the input graph and
inserting them into the appropriate hash tables. $\cT$ can be
initialized by iterating over edges $(u,w)$ in $\hl$ and
for each $w$, iterating over all edges $(w,v)$ in $\lh$ to find pairs
of high-degree vertices $u$ and $v$, and then incrementing $\cT(u,v)$.

\myparagraph{The Kara et al. Algorithm~\cite{KNNOZ19}}
Given an edge insertion $(u, v)$ (deletions are handled similarly, and for
simplicity assume that the edge does not already exist in $G$), the
update algorithm must identify all tuples $(u,w,v)$ where $(u,w)$ and
$(v,w)$ already exist in $G$, since such triples correspond to new
triangles formed by the edge insertion. The algorithm proceeds by
considering how a triangle's edges can reside in the data structures.
For example, if all of $u$, $v$, and $w$ are high-degree, then the algorithm
will enumerate these triangles by checking $\hh$ and finding all
neighbors $w$ of $u$ that are also high-degree (there are at most
$O(\sqrt{m})$ such neighbors), checking if the $(v,w)$ edge exists in
constant time.  On the other hand, if $u$ is low-degree, then checking
its $O(\sqrt{m})$ many neighbors suffices to enumerate all new
triangles.  The interesting case is if both $u$ and $v$ are
high-degree, but $w$ is low-degree, since there can be much more than
$O(\sqrt{m})$ such $w$'s. This case is handled using $\cT$, which
stores for a given $(u,v)$ edge in $\hh$ all $w$ such that $(w,u)$ and
$(w,v)$ both exist in $\lh$.

Finally, the algorithm updates the data structures, first inserting
the new edge into the appropriate edge-store. The algorithm updates
$\cT$ as follows. If $u$ and $v$ are both low-degree or both
high-degree, then no update is needed to $\cT$. Otherwise, without
loss of generality suppose $u$ is low-degree and $v$ is high-degree.
Then, the algorithm enumerates all high-degree vertices $w$ that are neighbors
of $u$ and increments the value of $(v,w)$ in $\cT$.

\ifFull
\subsection{Parallel Batch-Dynamic Update Algorithm}\label{sec:update-alg}
\fi

We present a high-level overview of our parallel algorithm in this section, and a more detailed description in Section~\ref{sec:triangle-full-alg}.
We consider batches of $\Delta$ edge insertions and/or deletions.  Let
$\ins(u, v)$ represent the update corresponding to inserting an edge
between vertices $u$ and $v$, and $\del(u, v)$ represent deleting the
edge between $u$ and $v$. We first preprocess the batch to account for
updates that \emph{nullify} each other. For example, an $\ins(u, v)$ update
followed chronologically by
a $\del(u, v)$ update nullify each other because the $(u, v)$ edge
that is inserted is immediately deleted, resulting in no change to the
graph. To process the batch consisting of nullifying updates, we claim
that the only update that is not nullifying for any pair of vertices
is the chronologically last update in the batch for that edge.
Since all updates contain a timestamp, to account for nullifying updates
we first find all updates on the same edge by hashing the updates
by the edge that it is being performed on. Then, we run the parallel
maximum-finding algorithm given in~\cite{Vishkin08}
on each set of updates for each edge in parallel.
This maximum-finding algorithm then
returns the update with the largest timestamp (the most recent update)
from the set of updates for each edge. This set of returned
updates then composes a batch of non-nullifying updates.

Before we go into the details of our parallel batch-dynamic triangle
counting algorithm, we first describe some challenges that must be
solved in using Kara et al.~\cite{KNNOZ19}
for the parallel batch-dynamic setting.

\myparagraph{Challenges}
Because Kara et al.~\cite{KNNOZ19} only considers one update
at a time in their algorithm, they do not deal with cases where
a set of two or more updates creates a new triangle. Since, in
our setting, we must account for batches of multiple updates,
we encounter the following set of challenges:

\begin{enumerate}[label=(\textbf{\arabic*}),topsep=1pt,itemsep=0pt,parsep=0pt,leftmargin=15pt]
    \item We must be able to efficient find
        new triangles that are created via
        two or more edge insertions.
    \item We must be able to handle insertions and deletions simultaneously
        meaning that a triangle with one inserted edge and one deleted
        edge should not be counted as a new triangle.
    \item We must account for over-counting of triangles due
        to multiple updates occurring simultaneously.
\end{enumerate}

For the rest of this section, we assume that $\Delta \leq m$, as
otherwise we can re-initialize our data structure using the static
parallel triangle-counting algorithm~\cite{ShunT2015}\footnote{The hashing-based
version of the algorithm given in~\cite{ShunT2015} can be modified to obtain the stated bounds
if it does not do ranking and
when using the $O(\log^* n)$ depth w.h.p.\ parallel hash table and uses atomic-add.} to get the count in
$O(\Delta^{3/2})$ work, $O(\log^* \Delta)$ depth, and $O(\Delta)$ space (assuming atomic-add),
which is within the bounds of Theorem~\ref{thm:linear-space-update}.

\myparagraph{Parallel Initialization}
Given a graph with $m$ edges, we initialize the triangle count $C$
using a static parallel triangle counting algorithm in $O(\alpha m)=O(m^{3/2})$
work, $O(\log^* m)$ depth, and $O(m)$ space~\cite{ShunT2015}, using atomic-add.
We initialize $\hh$, $\hl$, $\lh$, and $\lowl$ by scanning the edges
in parallel and inserting them into the appropriate parallel hash
tables. We initialize the degree array $\degarray$ by scanning the vertices.
Both steps take $O(m)$ work and $O(\log^*m)$ depth w.h.p.  $\cT$ can be
initialized by iterating over edges $(u,w)$ in $\hl$ in parallel and
for each $w$, iterating over all edges $(w,v)$ in $\lh$ in parallel to
find pairs of high-degree vertices $u$ and $v$, and then incrementing
$\cT(u,v)$.  The number of entries in $\hl$ is $O(m)$ and each $w$ has
$O(\sqrt{m})$ neighbors in $\lh$, giving a total of $O(m^{3/2})$ work
and $O(\log^* m)$ depth w.h.p.\ for the hash table insertions. The
amortized work per edge update is $O(\sqrt{m})$.

\myparagraph{Data Structure Modifications}
We now describe additional information that is stored in $\hh$, $\hl$,
$\lh$, $\lowl$, and $\cT$, which is used by the \batchdynamic{} update
algorithm:

\begin{enumerate}[label=(\textbf{\arabic*}),topsep=0pt,itemsep=0pt,parsep=0pt,leftmargin=15pt]
\item Every edge stored in $\hh$, $\hl$, $\lh$, and $\lowl$ stores an
  associated state, indicating whether it is an \defn{old edge}, a
    \defn{new insertion} or a \defn{new deletion}, which correspond to
    the values of 0, 1, and 2, respectively.
\item $\cT(u, v)$ stores a tuple with 5 values instead of a single value for each index $(u, v)$.
  Specifically, a $5$-tuple entry of $\cT(u, v) = (\tup{1},
    \tup{2}, \tup{3}, \tup{4}, \tup{5})$ represents the following:
  \begin{itemize}[topsep=0pt,itemsep=0pt,parsep=0pt,leftmargin=15pt]
      \item $\tup{1}$ represents the number of wedges with endpoints
        $u$ and $v$ that include only old edges.
      \item $\tup{2}$ and $\tup{3}$ represent the number of wedges
        with endpoints $u$ and $v$ containing one or two newly
        inserted edges, respectively.
      \item $\tup{4}$ and $\tup{5}$ represent the number of wedges
        with endpoints $u$ and $v$ containing one or two newly deleted
        edges, respectively. In other words, they are wedges that do
        not exist anymore due to one or two edge deletions.
  \end{itemize}
\end{enumerate}

\myparagraph{Algorithm Overview}
We first remove updates in the batch that either insert edges already
in the graph or delete edges not in the graph by using approximate compaction to filter.
Next, we update the tables $\hh$, $\hl$, $\lh$, and $\lowl$ with the new edge
insertions. Recall that we must update the tables with both $(u, v)$
and $(v, u)$ (and similarly when we update these tables with edge
deletions). We also mark these edges as newly inserted. Next, we
update $\degarray$ with the new degrees of all vertices due to edge
insertions. Since the degrees of some vertices have now increased, for low-degree vertices whose degree exceeds $t_2$, in
parallel, we promote them to high-degree vertices, which involves updating the tables $\hh$, $\hl$, $\lh$, $\lowl$, and $\cT$.
Next, we update the tables $\hh$, $\hl$, $\lh$, and $\lowl$ with new
edge deletions, and mark these edges as newly deleted. We then call
the procedures $\tinsnoarg$ and $\tdelnoarg$, which update the
wedge-table $\cT$ based on all new insertions and all new
deletions, respectively. At this point, our auxiliary data structures contain both
new triangles formed by edge insertions, and triangles deleted due to
edge deletions.

For each update in the batch, we then determine the number of new
triangles that are created by counting different types of triangles
that the edge appears in (based on the number of other updates forming
the triangle). We then aggregate these per-update counts to update the
overall triangle count.

Now that the count is updated, the remaining steps of the algorithm
handle unmarking the edges and restoring the data structures so that
they can be used by the next batch. We unmark all newly inserted edges
from the tables, and delete all edges marked as deletes in this batch.
Finally, we handle updating $\cT$, the wedge-table for all insertions
and deletions of edges incident to low-degree vertices.  The last
steps in our algorithm are to update the degrees in response to the newly inserted edges and the
now truly deleted edges. Then, since the degrees of some high-degree
vertices may drop below $t_1$ (and vice versa), we convert them to low-degree vertices
and update the tables $\hh$, $\hl$, $\lh$, $\lowl$, and
$\cT$ (and vice versa). This step is called \defn{minor rebalancing}. Finally, if the number of edges in the graph becomes less
than $M/4$ or greater than $M$ we reset the values of $M$, $t_1$, and
$t_2$, and re-initialize all of the data structures. This step is called \defn{major rebalancing}.

\myparagraph{Algorithm Description}
A simplified version of our algorithm is shown below.
The following $\textsc{Count-Triangle}$ procedure
takes as input a batch of $\Delta$ updates
$\mathcal{B}$ and returns the count of the updated number of triangles in the graph (assuming
the initialization process has already been run on the input graph and all associated data
structures are up-to-date).

\begin{mdframedalg}{Simplified parallel batch-dynamic triangle counting algorithm.}
    \label{alg:count-cliques}
\begin{algorithmic}[1]
    \Function{Count-Triangles}{$\batchset$}
    \ParFor{$\ins(u, v) \in \batchset$}
        \State Update and label edges $(u, v)$ and $(v, u)$ in $\hh$,
        \Statex \ \ \ \ \ \ \ \ \ \ \ \ $\hl$, $\lh$, and $\lowl$ as inserted edges.
    \EndParFor
    \ParFor{$\del(u, v) \in \batchset$}
        \State Update and label edges $(u, v)$ and $(v, u)$ in $\hh$,
        \Statex \ \ \ \ \ \ \ \ \ \ \ \ $\hl$, $\lh$, and $\lowl$ as deleted edges.
    \EndParFor
    \ParFor{$\ins(u, v) \in \batchset$ or $\del(u, v) \in \batchset$}
        \State Update $\cT$ with $(u, v)$. $\cT$ records the number of
        \Statex \ \ \ \ \ \ \ \ \ \ \ \ wedges that have $0$, $1$, or
        $2$ edge updates.
    \EndParFor
    \ParFor{$\ins(u, v) \in \batchset$ or $\del(u, v) \in \batchset$}
    \State Count the number of new triangles and deleted
    \Statex \ \ \ \ \ \ \ \ \ \ \ \ triangles incident to edge $(u, v)$, and account for
    \Statex \ \ \ \ \ \ \ \ \ \ \ \ duplicates.
    \EndParFor
     \State Rebalance data structures if necessary.
    \EndFunction
\end{algorithmic}
\end{mdframedalg}

\myparagraph{Small Example Batch Updates} Here we provide a small example of processing a batch of updates. We assume that no rebalancing occurs.
Suppose we have a
batch of updates containing an edge insertion $(u, v)$ with timestamp $3$,
an edge deletion $(w, x)$ with timestamp $1$, and an edge deletion $(u, v)$ with
timestamp $2$. Since the edge insertion $(u, v)$ has the later timestamp, it is
the update that remains. After removing nullifying updates, the two updates that
remain are insertion of $(u, v)$ and deletion of $(w, x)$. The algorithm first
looks in $\degarray$ to find the degrees of $u$, $v$, $w$, and $x$ in parallel. Suppose
$u$, $v$, and $w$ are high-degree and $x$ is low-degree.
We need to first update our data structures with the new edge updates.
To update the data structure, we first
update the edge table $\hh$ with $(u, v)$ marked as an edge insertion.
Then, we update the edge tables $\hl$ and $\lh$ with $(w, x)$ as an edge
deletion. Finally, we update the counts of
wedges in $\cT$ with $(w, x)$'s deletion. Specifically,
for each of $x$'s neighbors $y$ in $\lh$, we update $\cT(w, y)$ by incrementing
$t^{(w, y)}_4$ (since $(x, y)$ is not a new update).

After updating the data structures, we can count the changes to the total number
of triangles in the graph. All of the following actions can be performed in
parallel.
Suppose that $u$ comes lexicographically before $v$. We count the number of
neighbors of $u$ in $\hh$ and this will be the number of new triangles
containing three high-degree vertices. To avoid overcounting, we do not count
the number of high-degree neighbors of $v$. Since we are counting the number of
triangles containing updates, we also do not count the number of high-degree
neighbors of $w$ since $(w, x)$ cannot be part of any new triangles containing three
high-degree vertices.
Then, in parallel, we count the number of
neighbors of $x$ in $\lowl$ and $\lh$; this is the number of
deleted triangles containing
one and two high-degree vertices, respectively.
We use $\cT$ to count the number of
triangles containing one low-degree vertex and $(u, v)$. To count the number of
inserted triangles containing $(u, v)$ and a
low-degree vertex, we look up $t_1^{(u, v)}$
in $\cT$ and add it to our final triangle count; all
other stored count values for $(u, v)$ in $\cT$ are $0$
since there are no other new
updates incident to $u$ or $v$.

\ifFull
\subsection{Parallel Batch-Dynamic Triangle Counting Detailed Algorithm}\label{sec:triangle-full-alg}
\fi

The detailed pseudocode of our parallel \batchdynamic{} triangle
counting algorithm are shown below. Recall that the update procedure for a set of
$\Delta\leq m$ non-nullifying updates is as follows (the subroutines
used in the following steps are described afterward).

\begin{breakablealgorithm}{}\label{alg:batchdynamictriangle}
\caption{Detailed parallel batch-dynamic triangle counting procedure.}
\begin{enumerate}[label=(\textbf{\arabic*}),topsep=0pt,itemsep=0pt,parsep=0pt,leftmargin=20pt]
\item Remove updates that insert edges already in
  the graph or delete edges not in the graph as well as nullifying updates using approximate compaction.\label{bt:removebadupdates}

\item Update tables $\hh$, $\hl$, $\lh$, and $\lowl$ with the new
    edge insertions using $\ins(u, v)$ and $\ins(v, u)$.
  Mark these edges as newly inserted by running
  $\markins{\batchset}$ on the batch of updates
  $\batchset$.\label{bt:updateedgetablesinsert}

\item Update tables $\hh$, $\hl$, $\lh$, and $\lowl$ with new
    edge deletions using $\del(u,v)$ and $\del(v, u)$.
  Mark these edges as newly
  deleted using $\markdels{\batchset}$ on
  $\batchset$.\label{bt:updateedgetablesdelete}

\item Call $\tins{\batchset}$ for the set $\batchset$ of all edge
  insertions $\ins(u, w)$, where either $u$ or $w$ is
  low-degree and the other is high-degree.
  \label{bt:updatetableinsertions}

\item Call $\tdel{\batchset}$ for the set $\batchset$ of all edge
  deletions $\del(u, w)$ where either $u$ or $w$ is
  low-degree and the other is high-degree.
  \label{bt:updatetabledeletions}

\item For each update in the batch, determine the number of new
  triangles that are created by counting 6 values.
  Count the values using a 6-tuple, $(c_1, c_2, c_3, c_4, c_5, c_6)$
  based on the number of other updates contained in a
  triangle:\label{bt:countnewtrianglesperupdate}

\begin{enumerate}[topsep=0pt,itemsep=0pt,parsep=0pt,leftmargin=20pt]
    \item For each edge insertion $\ins(u, v)$ resulting in a triangle containing only one newly inserted edge (and no deleted edges), increment $c_1$ by $\ctriang{1}{0}{\ins(u, v)}$.
    \item For each edge insertion $\ins(u, v)$ resulting in a triangle containing two newly inserted edges (and no deleted edges), increment $c_2$ by $\ctriang{2}{0}{\ins(u, v)}$.
\item For each edge insertion $\ins(u, v)$ resulting in a triangle containing three newly inserted edges, increment $c_3$ by $\ctriang{3}{0}{\ins(u, v)}$.
\item For each edge deletion $\del(u, v)$ resulting in a deleted triangle with one newly deleted edge, increment $c_4$ by $\ctriang{0}{1}{\del(u, v)}$.
\item For each edge deletion $\del(u, v)$ resulting in a deleted triangle with two newly deleted edges, increment $c_5$ by $\ctriang{0}{2}{\del(u, v)}$.
\item For each edge deletion $\del(u, v)$ resulting in a deleted triangle with three newly deleted edges, increment $c_6$ by $\ctriang{0}{3}{\del(u, v)}$.
\end{enumerate}

Let $C$ be the previous count of the number of triangles.
Update $C$ to be $C + c_1 + (1/2)c_2 + (1/3)c_3 - c_4 - (1/2)c_5 -
(1/3)c_6$, which becomes the new count.

\item Scan through updates again. For each
  update, if the value stored in $\hh$, $\hl$, $\lh$,
  and/or $\lowl$ is $2$ (a deleted edge), remove this edge. If
  stored value is $1$ (an inserted edge), change the value to $0$. For all
  updates where the endpoints are both high-degree or both
  low-degree, we are done. For each update $(u, w)$ where
  either $u$ or $w$ is low-degree (assume without loss of generality that $w$ is)
  and the other is high-degree, look for all
  high-degree neighbors $v$ of $w$ and update $\cT(u, v)$ by summing
  all $c_1$, $c_2,$, and $c_3$ of the tuple and subtracting $c_4$ and $c_5$.

  \label{bt:scanthroughupdatesagain}

  \item Update $\degarray$ with the new degrees.\label{bt:updatedegree}

  \item Perform minor rebalancing for all vertices
  $v$ that exceed $t_2$ in degree or fall under $t_1$
  in parallel using
  $\minreb{v}$. This makes a formerly low-degree vertex
  high-degree (and vice versa) and updates relevant structures.\label{bt:performminorrebalancing}

\item Perform major rebalancing if necessary (i.e., the total number
  of edges in the graph is less than $M/4$ or greater than $M$).
  Major rebalancing re-initializes all
    structures.\label{bt:majorrebalancing}
\end{enumerate}
\end{breakablealgorithm}

\mysubsection{Procedure $\markins{\batchset}$} We scan through each of
the $\ins(u, v)$ updates in $\batchset$ and mark $(u, v)$ and $(v, u)$
as newly inserted edges in $\hh$, $\hl$, $\lh$, and/or $\lowl$ by
storing a value of $1$ associated with the edge.

\mysubsection{Procedure $\markdels{\batchset}$} Because we removed all
nullifying updates before $\batchset$ is passed into the procedure,
none of the deletion updates in $\batchset$ should delete newly inserted
edges. For all edge deletions $\del(u, v)$, we change the values
stored under $(u, v)$ and $(v, u)$ from $0$ to $2$ in the tables
$\hh$, $\hl$, $\lh$, and/or $\lowl$.

\mysubsection{Procedure $\tins{\batchset}$} For each $(u,w) \in
\batchset$, assume without loss of generality that $w$ is the
low-degree vertex and do the following.  We first find all of $w$'s
neighbors, $v$, in $\lh$ in parallel. Then, we determine for each
neighbor $v$ if $(w, v)$ is new (marked as $1$). If the edge $(w, v)$
is not new, then increment the second value stored in the tuple with
index $\cT(u, v)$. If $(w, v)$ is newly inserted, then increment the third value
stored in $\cT(u, v)$. The first, fourth, and fifth values stored in
$\cT(u, v)$ do not change in this step. The first, second, and third values
count the number of edge insertions contained in the wedge keyed by $(u, v)$.
The first value counts all wedges with endpoints $u$ and $v$ that do not
contain any edge update, the second count the number of wedges containing one
edge insertion, and the third counts the number of wedges containing two edge
insertions.
Then, intuitively, the first,
second, and third values will tell us later for edge insertion $(u, v)$
between two high-degree vertices
whether newly created
triangles containing $(u, v)$
have one (the only update being $(u, v)$), two, or three,
respectively, new edge
insertions from the batch update.
We do not update the edge insertion counts of
wedges which contain a mix of edge insertion updates
and edge deletion updates.

\mysubsection{Procedure $\tdel{\batchset}$}  For each $(u,w) \in
\batchset$, assume without loss of generality that $w$ is the
low-degree vertex and do the following. We first find all of $w$'s
neighbors, $v$, in $\lh$ in parallel. Then, we determine for each
neighbor $v$ if $(w, v)$ is a newly deleted edge (marked as $2$). If
$(w, v)$ is not a newly deleted edge, increment the fourth value in
the tuple stored in $\cT(u, v)$ and decrement the first value.
Otherwise, if $(w, v)$ is a newly
deleted edge, increment the fifth value of $\cT(u, v)$ and decrement
the first value. The
second and third values in $\cT(u, v)$ do not change in this step.
For any key $(u, v)$, the first, fourth, and fifth values gives the number of
wedges with endpoints $u$ and $v$ that contain zero, one, or two edge deletions,
respectively.
Intuitively, the first, fourth, and fifth values tell us later
whether newly deleted triangles  have one (where the only edge deletion is
$(u, v)$), two, or three,
respectively, new edge deletions from the batch update.

\mysubsection{Procedure $\ctriang{i}{d}{update}$} This procedure
returns the number of triangles containing the update $\ins(u, v)$ or
$\del(u, v)$ and exactly $i$ newly inserted edges or exactly $d$ newly
deleted edges (the update itself counts as one newly inserted edge or
one newly deleted edge). If at least one of $u$ or $v$ is low-degree,
we search in the tables, $\lh$, and $\lowl$ for neighbors
of the low-degree vertex and the number of
marked edges per triangle: edges marked as $1$ for insertion updates and edges
marked as $2$ for deletion updates.
If both $u$ and $v$ are
high-degree, we first look through all high-degree
vertices using $\hh$ to see if any
form a triangle with both high-degree endpoints $u$ and $v$ of the update.
This allows us to find all newly updated triangles containing only
high-degree vertices.
Then, we confirm the existence
of a triangle for each neighbor found in the tables by checking for the third edge
in $\hh$, $\hl$, $\lh$, or $\lowl$. We return only the counts containing the
correct number of updates of the correct type.
To avoid double counting for each update we
do the following. Suppose all vertices are ordered lexicographically in some
order.
For any edge which contains two high-degree or two low-degree
vertices, we search in $\lowl$, $\hh$, and $\lh$ for exactly one of the two
endpoints, the one that is lexicographically smaller.

Then, we return a tuple in $\cT(u, v)$ based on the values of $i$ and $d$
to
determine the count of triangles containing $u$ and $v$ and one low-degree vertex:
\begin{itemize}[topsep=0pt,itemsep=0pt,parsep=0pt,leftmargin=20pt]
\item Return the first value $\tup{1}$ if either $i = 1$ or $d = 1$.
\item Return the second value $\tup{2}$ if $i = 2$.
\item Return the third value $\tup{3}$ if $i = 3$.
\item Return the fourth value $\tup{4}$ if $d = 2$.
\item Return the fifth value $\tup{5}$ if $d = 3$.
\end{itemize}

Note that we ignore all triangles that include more than one
insertion update \emph{and} more than one deletion update.

\mysubsection{Procedure $\minreb{u}$} This procedure performs a minor rebalance when either the degree of $u$ decreases below $t_1$ or increases above $t_2$. We move all edges in $\hh$ and $\hl$ to $\lh$ and $\lowl$ and vice versa.
We also update $\cT$ with
new pairs of vertices that became high-degree and delete pairs
that are no longer both high-degree.

\ifFull
    \subsection{Analysis}
\fi

We prove the correctness of our algorithm in the following
theorem. The proof is based on accounting for the
contributions of an edge to each triangle that it participates in based on
the number of other updated edges found in the triangle.

\begin{restatable}{theorem}{batchtriangcorrect}\label{thm:batch-triang-correct}
Our parallel \batchdynamic{} algorithm maintains the number of triangles in the graph.
\end{restatable}

\begin{proof}
All triangles containing at least one low-degree vertex can be found
either in $\cT$ or by searching through $\lh$ and $\lowl$. All
triangles containing all high-degree vertices can be found by
searching $\hh$.  Suppose that an edge update $\ins(u, v)$
(resp.\ $\del(u,v)$) is part of $I_{(u,v)}$ (resp.\ $D_{(u,v)}$) triangles. We need
to add or subtract from the total count of triangles $I_{(u,v)}$ or
$D_{(u,v)}$, respectively. However, some of the triangles will be
counted twice or three times if they contain more than one edge
update. By dividing each triangle count by the number of updated edges
they contain, each triangle is counted exactly once for the total
count $C$.
\end{proof}

\myparagraph{Overall Bound}
We now prove that our parallel \batchdynamic{} algorithm runs in $O(\batch
\sqrt{\batch+m})$ work, $O(\log^*(\batch + m))$ depth, and uses
$O(\batch+m)$ space. Henceforth, we assume that our algorithm
uses the atomic-add instruction (see Section~\ref{sec:prelims}).
Removing nullifying updates takes $O(\batch)$
total work, $O(\log^*{\batch})$ depth w.h.p., and
$O(\batch)$ space for hashing and the find-maximum procedure outlined
in Section~\ref{sec:update-alg}.  In
step~\ref{bt:removebadupdates}, we perform table lookups for the
updates into $\degarray$ and
in $\hh$, $\hl$, $\lh$, or $\lowl$,
followed by approximate compaction to filter. The hash table lookups take $O(\batch)$
work and $O(\log^*m)$ depth with high probability and $O(m)$ space.
Approximate compaction~\cite{Gil91a} takes $O(\batch)$ work, $O(\log^*\Delta)$ depth, and
$O(\batch)$ space.
Steps~\ref{bt:updateedgetablesinsert},~\ref{bt:updateedgetablesdelete}, and~\ref{bt:updatedegree} perform hash table insertions and updates on the batch of
$O(\batch)$ edges, which takes
$O(\batch)$ amortized work and $O(\log^*m)$ depth with high
probability.

The next lemma shows that updating the tables based on the edges in
the update (steps~\ref{bt:updatetableinsertions} and
\ref{bt:updatetabledeletions}) can be done in $O(\batch\sqrt{m})$ work
and $O(\log^*m)$ depth w.h.p., and $O(m)$ space.

\begin{restatable}{lemma}{updatetableins}\label{lem:update-table-ins}
    $\tins{\batchset}$ and $\tdel{\batchset}$ on a batch $\batchset$ of size $\batch$ takes $O(\batch\sqrt{m})$ work and $O(\log^* (\batch + m))$ depth w.h.p., and $O(\batch+m)$ space.
\end{restatable}

\begin{proof}
  For each $w$, we find all of its high-degree neighbors in $\lh$ and perform the increment or decrement in the corresponding entry in $\cT$ in parallel (at this point, the vertices are still classified based on their original degrees). The total number of new neighbors gained across all  vertices is $O(\batch)$ since there are $\batch$ updates.
  Therefore,  across all updates, this takes $O(\batch\sqrt{m}+\batch)$ work and $O(\log^*{(\batch+m)})$ depth w.h.p. due to hash table lookup and updates.
  Then, for all high-degree neighbors found, we perform the increments or decrements on the corresponding entries in $\cT$ in parallel, taking the same bounds.
All vertices can be processed in parallel, giving a total of $O(\batch\sqrt{m}+\batch)$  work and $O(\log^*(\batch+m))$ depth w.h.p.
\end{proof}

The next lemma bounds the complexity of updating the triangle count
in step~\ref{bt:countnewtrianglesperupdate}.

\begin{restatable}{lemma}{updatingtriangles}\label{lem:updating-triangles}
Updating the triangle count takes $O(\batch\sqrt{m})$ work and
$O(\log^* (\batch+m))$ depth w.h.p., and $O(\batch+m)$ space.
\end{restatable}

\begin{proof} We initialize $c_1,\ldots,c_6$ to $0$.
  For each edge update in $\batchset$ where both endpoints are
  high-degree, we perform lookups in $\cT$ and $\hh$ for the relevant
  values in parallel and increment the appropriate $c_i$. Finding all
  triangles containing the edge update and containing only high-degree
  vertices takes $O(\batch \sqrt{m})$ work and $O(\log^* (\batch+m))$
  depth w.h.p. This is because there are $O(\sqrt{m})$ high-degree vertices in
  total, and for each we check whether it appears in the $\hh$ table for both
  endpoints of each update.
  Performing lookups in $\cT$ takes $O(\batch)$ work and
  $O(\log^* (\batch+m))$ depth w.h.p.

  For each update containing at
  least one endpoint with low-degree, we perform lookups in the tables
  $\hl$, $\lh$, and $\lowl$ to find all triangles containing the
  update and increment the appropriate $c_i$. This takes $O(\batch
  \sqrt{m}+\batch)$ work and $O(\log^*(\batch+m))$ depth
  w.h.p. Incrementing all $c_i$'s for all newly updated triangles
  takes $O(\batch)$ work and $O(1)$ depth.  We then apply the equation
  in step~\ref{bt:countnewtrianglesperupdate} to update $C$, which
  takes $O(1)$ work and depth.
\end{proof}

The following lemma
bounds the cost for minor rebalancing, where a
low-degree vertex becomes high-degree or vice versa
(step~\ref{bt:performminorrebalancing}).

\begin{restatable}{lemma}{minorrebalance}\label{lem:minor-rebalancing}
Minor rebalancing for edge updates
  takes $O(\batch \sqrt{m})$ amortized work and $O(\log^*(\batch+ m))$ depth w.h.p., and $O(\batch+m)$ space.
\end{restatable}
\begin{proof}
We describe the case of edge insertions, and the case for edge
deletions is similar.  Using approximate compaction to perform the
filtering, we first find the set $S$ of low-degree vertices exceeding
$t_2$ in degree. This step takes $O(\batch)$ work and $O(\log^*
\batch)$ depth w.h.p. For vertices in $S$, we then delete the edges from
their old hash tables and move the edges to their new hash tables. The
work for each vertex is proportional to its current degree, giving a
total work of $O(\sum_{v\in S}\deg(v)) = O(\batch\sqrt{m}+\batch)$
w.h.p. since the original degree of low-degree vertices is
$O(\sqrt{m})$ and each edge in the batch could have caused at most 2
such vertices to have their degree increase by 1 (the w.h.p.\ is for
parallel hash table operations).

In addition to moving the edges into new hash tables, we also have to
update $\cT$ with new pairs of vertices that became high-degree and
delete pairs of vertices that are no longer both high-degree.
To
update these tables, we need to find all new pairs of high-degree
vertices. There are at most $O(\Delta \sqrt{m+\batch})$ such new pairs, which
can be found by filtering neighbors using approximate compaction of vertices in $S$ in
$O(\Delta\sqrt{m+\batch})$ work and $O(\log^*(\batch + m))$ depth w.h.p.
For each pair $(u,v)$, we
check all neighbors of an endpoint
that just became high-degree and
increment the entry $\cT(u, v)$ for each low-degree neighbor $w$ found
that has edges $(u, w)$ and $(w, v)$.
Low-degree
neighbors have degree $O(\sqrt{m+\batch})$, and so the total work is
$O(\batch (m+\batch))$ and depth is $O(\log^*(\batch + m))$ w.h.p.
using atomic-add.  There must have been
$\Omega(\sqrt{m})$ updates on a vertex before minor rebalancing is
triggered, and so the amortized work per update is $O(\Delta
\sqrt{m})$ and the depth is $O(\log^* m)$ w.h.p. The space for filtering is
$O(m+\batch)$.
\end{proof}

We now finish showing Theorem~\ref{thm:linear-space-update}.
Lemma~\ref{thm:batch-triang-correct} shows that our algorithm
maintains the correct count of triangles.
Lemmas~\ref{lem:update-table-ins},~\ref{lem:updating-triangles}, and~\ref{lem:minor-rebalancing}  show that the cost of
updating tables to reflect the batch, updating the triangle counts, and minor rebalancing
is  $O(\batch\sqrt{m}+\batch)$ amortized work and $O(\log^*(\batch+m))$ depth w.h.p., and
$O(\batch+m)$ space.

Step~\ref{bt:scanthroughupdatesagain} can be done in
$O(\batch\sqrt{m})$ work and $O(\log^* m)$ depth as follows. We scan
through the batch $\batchset$ in parallel and update the hash tables
$\hh$, $\hl$, $\lh$, and $\lowl$ in $O(\batch)$ work and $O(\log^* (\batch+m))$
depth w.h.p.  For all updates in $\batchset$ containing one
high-degree vertex and one low-degree vertex, we update the table
$\cT$ in parallel by scanning the  neighbors in $\lh$ of
the low-degree vertex. This step takes $O(\batch \sqrt{m}+\batch)$ work and
$O(\log^* (\batch+m))$ depth w.h.p. Major rebalancing
(step~\ref{bt:majorrebalancing}) takes $O((\batch+m)^{3/2})$ work and $O(\log^*(\batch+
m))$ depth by re-initializing the data structures. The rebalancing
happens every $\Omega(m)$ updates, and so the amortized work per
update is $O(\sqrt{\batch+m})$ and depth is $O(\log^* (\batch+m))$
w.h.p.

Therefore, our update algorithm takes $O(\Delta\sqrt{\Delta+m})$ amortized work
and $O(\log^* (\batch+m))$ depth w.h.p., and $O(\batch+m)$ space overall
using atomic-add as stated in
Theorem~\ref{thm:linear-space-update}.

\myparagraph{Bounds without Atomic-Add}
Without the atomic-add
instruction, we can use a parallel reduction~\cite{JaJa92} to sum over
values when needed.
This is work-efficient and takes logarithmic depth, but uses space proportional to the number of values summed over in parallel.
For updates, this is
bounded by $O(\Delta\sqrt{m}+\Delta)$, and for initialization and major rebalancing, this is bounded by $O(\alpha m)$~\cite{ShunT2015}.
This would give
an overall bound of $O(\Delta(\sqrt{\Delta+m}))$ work and  $O(\log
(\Delta+m))$ depth w.h.p., and $O(\alpha m+\Delta\sqrt{m})$ space.

\ifFull
\section{Dynamic $k$-Clique Counting via Fast Static Parallel Algorithms}
\label{sec:arboricityclique}

In this section, we present a very simple algorithm for dynamically
maintaining the number of $k$-cliques for $k > 3$ based on statically
enumerating a number of smaller cliques in the graph, and intersecting
the enumerated cliques with the edge updates in the input batch.
Importantly, the algorithm is space-efficient, and only relies on
simple primitives such as clique enumeration of cliques of size smaller
than $k$, for which there are highly efficient algorithms both in
theory and practice.

\myparagraph{Fast Static Parallel $k$-Clique Enumeration}
The main tool used by algorithm is the following theorem, which is
presented in concurrent and independent work~\cite{shi2020parallel}:
\begin{theorem}[Theorem 4.2 of~\cite{shi2020parallel}]\label{thm:static_parallel_enumerate}
There is a parallel algorithm that given a graph $G$ can enumerate all
$k$-cliques in $G$ in $O(m\alpha^{k-2})$ expected work and
$O(\log^{k-2} n)$ depth w.h.p., using $O(m)$ space.
\end{theorem}
Theorem~\ref{thm:static_parallel_enumerate} is proven by modifying the
Chiba-Nishizeki (CN) algorithm in the parallel setting, and combining
the CN algorithm with parallel low-outdegree orientation
algorithms~\cite{barenboim2010, Goodrich11}.

\myparagraph{A Dynamic $k$-Clique Counting Algorithm}
Given Theorem~\ref{thm:static_parallel_enumerate}, one approach to
maintain the number of $k$-cliques in $G$ upon receiving a batch of
insertions or deletions $\mathcal{B}$ is to have each edge $e$ in the batch
simply enumerate all $(k-2)$-cliques, check whether $e$ forms a
$k$-clique with any of these $(k-2)$-cliques, and update the clique
counts based on the newly discovered (or deleted) cliques.

Algorithm~\ref{alg:arboricity_dynamic_count} presents a formalized
version of this idea. The algorithm first removes all nullifying
updates from $\mathcal{B}$. It then checks whether the batch is
large ($\Delta \geq m$), and if so simply recomputes the overall
$k$-clique count by re-running the static enumeration algorithm.
Otherwise, the algorithm inserts the edge insertions in the batch into
$G$, and stores them in a static parallel hash table $\mathcal{H}$
that maps each edge in the batch to a value indicating whether the
edge is an insertion or deletion in $\mathcal{B}$.

\begin{mdframedalg}{Dynamic $k$-Clique Counting}\label{alg:arboricity_dynamic_count}

  \begin{algorithmic}[1]
  \Function{$k$-Clique-Count}{$G=(V,E), \mathcal{B}$}
  \State Let $N$ be the current count of cliques before processing the current batch.
  \State Remove nullifying updates from $\mathcal{B}$.
  \If{$\Delta \geq m$}
    \State Rerun the static $k$-clique counting algorithm.
  \Else
    \State Insert all updates that are edge insertions in $\mathcal{B}$ into $G$.
    \State Let $\mathcal{H}$ be a static parallel hash table
    representing $\mathcal{B}$.
    \ParFor {$e=\{u,v\} \in \mathcal{B}$}
    \State Enumerate all $(k-2)$-cliques in $G$ in parallel using the Algorithm
    from Theorem~\ref{thm:static_parallel_enumerate}.
    \ParFor{each enumerated $(k-2)$-clique, $C$}
      \If {$C$ forms a newly inserted or newly deleted $k$-clique with
      $e$\label{line:checknewclique}}
        \If{$e=(u,v)$ is the lexicographically-first edge in $C$ in the
        batch}
         \State Atomically update the $k$-clique count with $C \cup \{u,v\}$: $N \leftarrow N + 1$.

        \EndIf
      \EndIf
    \EndParFor
  \EndParFor
  \State Delete all updates that are edge deletions in $\mathcal{B}$ from $G$.
  \EndIf
\EndFunction
\end{algorithmic}
\end{mdframedalg}

Then, in parallel, for each edge $e=(u,v)$ in the batch, it enumerates
all $(k-2)$-cliques in the graph. For each $(k-2)$-clique, $C$, the
algorithm checks whether this clique forms a newly inserted or newly
deleted $k$-clique with $e$. A newly inserted $k$-clique is one where
at least one edge is an edge insertion in $\mathcal{B}$ and all other
edges are not deleted in $\mathcal{B}$.  Similarly a newly deleted
$k$-clique is one where at least one edge is an edge deletion in
$\mathcal{B}$ and all other edges are not edge insertions in
$\mathcal{B}$.  This step is done by querying the static parallel hash
table $\mathcal{H}$ for each edge in the clique to check whether it is
an insertion or deletion in $\mathcal{B}$. Cliques consisting of a mix
of edge insertions and deletions are cliques that are not previously
present before the batch, and will not be present after the batch, and
are thus ignored.

For a newly inserted or deleted clique, the algorithm then
checks whether $e$ is the {\em lexicographically-first edge in the
batch} inside of this clique formed by $C \cup \{u,v\}$ (otherwise, a
different edge update from the batch will find and handle the
processing of this clique).\footnote{An edge $e=(u,v)$ is the lexicographically
        first edge in the batch in a clique $C$ if, $\forall e' =
        (u',v') \in C$ such that $(u',v') \in \mathcal{B}$, $e$ is
        lexicographically smaller than $e'$. Note that we are
        working over an undirected graph without self-loops. By
        convention, when discussing lexicographic comparison, we have
        that for any $e=(u,v)$ that $u < v$;
        in other words, the order in the tuple representing
        the edge is based on the lexicographical order of the two endpoints.}  Checking whether $e$ is the
lexicographically-first edge in a clique $C$ is done by querying the
static parallel hash table $\mathcal{H}$. For each clique where $e$ is
the lexicographically-first edge in the batch in the clique, we either
atomically increment, or decrement the count, based on whether this
clique is newly inserted or newly deleted. After the clique count has
been updated, the algorithm updates $G$ by performing the edge
deletions from $\mathcal{B}$.

We note that we could just as well enumerate all of the
$(k-2)$-cliques a single time, and then for each $(k-2)$-clique we
discover, check whether it forms a $k$-clique with each edge in the
batch. A practical optimization of this idea may store edges in a
batch incident to their corresponding endpoints, and so vertices in
the discovered $(k-2)$-clique would only need to check updates
incident to the vertices in this clique. The asymptotic
complexity of both ideas---joining cliques with edges, instead of
edges with cliques, and pruning edges from the batch to consider---is
the same in the worst case.

\myparagraph{Correctness and Bounds}
If a $k$-clique in the graph is not incident to any edges in the
batch, then its count is unaffected (since we only perform
modifications to the count for cliques containing edges in
$\mathcal{B}$). For cliques incident to edges in $\mathcal{B}$, we
consider two cases. If the clique $C$ is deleted after applying
$\mathcal{B}$, observe that by decomposing $C$ into a $(k-2)$-clique and
the lexicographically-first marked edge $e$ in $C$, $C$ will be found
and counted by $e$. The argument that a newly inserted clique, $C$,
will be found is similar. Lastly, cliques consisting of both edge
insertions and deletions in $\mathcal{B}$ will be correctly ignored by
the check on Line~\ref{line:checknewclique}. In other words, we check in
parallel whether any enumerated $k$-clique $C \cup \left\{u, v\right\}$ contains
both an edge deletion and an edge insertion (by checking in the hash table
representing $\batchset$);
if so, the $k$-clique composed of $C \cup \left\{u,
v\right\}$ is not counted.
This argument proves the
following theorem:
\begin{theorem}\label{thm:arboricity_dynamic_count_correct}
  Algorithm~\ref{alg:arboricity_dynamic_count} correctly maintains the
  number of $k$-cliques in the graph.
\end{theorem}

\begin{theorem}\label{thm:arboricity_dynamic_count_bound}
  Given a collection of $\Delta$ updates, there is a batch-dynamic
  $k$-clique counting algorithm that updates the $k$-clique counts
  running in $O(\Delta(m+\Delta)\alpha^{k-4})$ expected work and
  $O(\log^{k-2} n)$ depth w.h.p., using $O(m + \batch)$
  space.
\end{theorem}
\begin{proof}
We analyze Algorithm~\ref{alg:arboricity_dynamic_count}. First,
updating the graph, assuming that the edges incident to each vertex are
represented sparsely using a parallel hash table, requires $O(\Delta)$
work and $O(\log^{*} n)$ depth w.h.p.

If $\Delta \geq m$, the algorithm calls the static $k$-clique counting
algorithm, which takes $O((m + \Delta)\alpha^{k-2})$ expected work.
Since $m = O(\Delta)$ and $\alpha^{2} = O(m + \Delta)$, the work of
calling the static algorithm is upper-bounded by $O(\Delta(m+\Delta)
\alpha^{k-4})$ as required.  Finally, the depth bound is
$O(\log^{k-2} n)$ w.h.p.\ as required.

Otherwise, $\Delta < m$. Then, the algorithm first inserts and marks
the batch in the graph.
It also stores the edges in the batch in a parallel hash
table. Creating the parallel hash table takes $O(\Delta)$ work and
$O(\log^{*} n)$ depth w.h.p., which are both subsumed by the overall
work and depth for the relevant setting of $k > 2$. For each update,
we list all $(k-2)$-cliques using the algorithm from
Theorem~\ref{thm:static_parallel_enumerate}.  This step can be done
in $O((m+\Delta)\alpha^{k-4})$ expected work and $O(\log^{k-4} n)$
depth w.h.p. If the
$(k-2)$-clique $C$ forms a $k$-clique with $e$, then the cost of checking
whether the clique is newly inserted or newly deleted using
$\mathcal{H}$ costs $O(k)$ work, which is a constant, and $O(1)$
depth. The cost of checking whether $e$ is the lexicographically first
edge in $\mathcal{B}$ is also constant.
Multiplying the cost of
enumeration by the number of edges in the batch completes the proof.
\end{proof}

Our batch-dynamic algorithm outperforms re-computation using the
static parallel $k$-clique counting algorithm for $\Delta = o(\alpha^{2})$.

It is an interesting open question whether our dependence
on $m$ could be entirely removed from the update bound.  Existing work
has provided efficient sequential dynamic algorithms maintaining the
$k$-clique count in $\tilde{O}(\alpha^{k-2})$ work per update
using dynamic low out-degree orientations~\cite{Dvorak2013}. It would
be interesting to understand whether such an algorithm can be
work-efficiently parallelized  in the parallel batch-dynamic setting, which would allow the dynamic algorithm to match
the work of static  parallel recomputation up to logarithmic factors.

\newcommand{\counttriangles}[1]{\mathtt{count\_updated\_low\_degree\_triangles}(#1)}

\section{Dynamic $k$-Clique via Fast Matrix Multiplication}\label{sec:mm}

In this section, we present our final result which is
a parallel \batchdynamic{} algorithm for
counting $k$-cliques based on fast matrix multiplication in general graphs
(which may be dense). For bounded arboricity graphs, we can also
count cliques in
$O(\Delta(m + \Delta)\alpha^{k-4})$ expected work and
$O(\log^{k-2} n)$ depth w.h.p., using $O(m + \Delta)$ space.
Due to the similarity of this result to the static parallel $k$-clique
counting algorithm given in~\cite{shi2020parallel}, we do not present the
details of the proof of this result here but instead refer the
interested reader to Appendix~\ref{sec:arboricityclique}.

Using
parallel matrix multiplication (discussed in Section~\ref{sec:pmm}), we achieve a better work
bound (in terms of $m$) for large values of $k$ than our bound of
$O(\Delta(\Delta+m)\alpha^{k-4})$ obtained from the simple algorithm presented in
Section~\ref{sec:arboricityclique}.
To the best of our knowledge, our algorithm (when made sequential) also achieves the
best runtime for any sequential dynamic $k$-clique counting algorithm on dense
graphs for large $k$ when using the best currently known matrix multiplication algorithm~\cite{Williams12,LeGall14}.
For values of $k > 9$, our MM based algorithm achieves $o(m^{k/2 - 1})$ amortized time compared to the arboricity-based
algorithm of~\cite{Dvorak2013} that dynamically counts cliques in $\tilde{O}(\alpha^{k-2})$ amortized time
where $\alpha$ is the arboricity of the graph (or $\tilde{O}\left(m^{k/2-1}\right)$ amortized time when
$\alpha = \Omega\left(\sqrt{m}\right)$) or the trivial
$O\left(m^{k/2-1}\right)$ algorithm of choosing all $k/2 - 1$ combinations of edges containing
neighbors of the incident vertices of the inserted edge.

Our dynamic algorithm modifies the algorithm of~\cite{AYZ97} for
counting triangles based on fast matrix multiplication and combines it
with a dynamic version of the static $k$-clique counting algorithm
of~\cite{EG04} to count the number of $k$-cliques under edge updates
in batches of size
$\Delta$. Sections~\ref{sec:kmod3}--\ref{sec:matrix-analysis} proves
the following theorem for the case when $k \bmod 3 =
0$. Section~\ref{sec:all-k-alg} describes the changes needed for the
case when $k \bmod 3 \neq 0$.

\begin{theorem}\label{thm:mm-main}
There exists a parallel \batchdynamic{} algorithm for counting the
number of $k$-cliques, where $k\bmod 3=0$, that takes
$O\left(\min\left(\Delta m^{\frac{(2k - 3)\op}{3(1+\op)}}, (m +
\batch)^{\frac{2k\op}{3(1+\op)}}\right)\right)$ amortized work and $O(\log (m + \batch))$
depth w.h.p., in $O\left((m + \batch)^{\frac{2k\op}{3(1+\op)}}\right)$ space, given a parallel
matrix multiplication algorithm with exponent $\op$.
\end{theorem}

Using the best currently known matrix multiplication algorithms with exponent $\op = 2.373$,
we obtain the following work and space bounds.

\begin{corollary}\label{cor:strassen-ws}
There exists a parallel \batchdynamic{} algorithm for counting the
number of $k$-cliques, where $k\bmod 3=0$, which takes
$O\left(\min(\batch m^{0.469k - 0.704}, (m +
\batch)^{0.469k})\right)$ work and $O(\log (m + \batch))$ depth w.h.p., in
$O\left((m + \batch)^{0.469k}\right)$ space by Corollary~\ref{cor:matrix-exp}.
\end{corollary}

Specifically, when amortized over the total number of edge updates
$\batch$, we obtain an amortized work bound of $O(m^{0.469k - 0.704})$ per edge update which is asymptotically better than
the combinatorial bound of $O\left(m^{k/2 - 1}\right)$ per update for
$k > 9$. To the best of our knowledge, this is also the best known worst-case
bound for dense graphs in the sequential setting.

Observe that our update algorithm only needs to handle batches of size $0 <
\batch \leq m^{\op/(1+\op)}$.  For batches
which have size $\batch > m^{\op/(1+\op)}$, we can reinitialize our
data structures in $O((m + \batch)^{0.469k})$ work ($O\left(m^{0.469k - 0.704}\right)$
amortized work per update in the batch), $O(\log \batch)$
depth, and $O((m + \batch)^{0.469k})$ space using our initialization algorithm described in
Lemma~\ref{lem:mmpreprocessing} and the fast parallel matrix multiplication of Corollary~\ref{cor:matrix-exp}, which is faster than using
the update algorithm (in general, we can use any fast matrix
multiplication algorithm that has low depth, but the cutoff for when
to reinitialize would be different). The analysis of the
reinitialization procedure (similar to the static case presented by Alon,
Yuster, and Zwick~\cite{AYZ97}) is provided in Section~\ref{sec:matrix-analysis}.  Thus, in the following sections,
we only describe our dynamic update procedures for batches of size $0
< \batch \leq m^{\op/(1+\op)}$.

\ifCameraReady
\subsection{Our Algorithm}\label{sec:kmod3}~
\fi
\ifFull
\subsection{Our Algorithm}\label{sec:kmod3}
\fi

In what follows, we assume that $k \bmod 3 = 0$ (please refer to
Section~\ref{sec:all-k-alg} for $k \bmod 3 \neq 0$). We use a batch-dynamic
triangle counting algorithm as a subroutine for our batch-dynamic $k$-clique
algorithm. Our algorithm for maintaining triangles is a batch-dynamic
version of the triangle counting algorithm by Alon, Yuster, and Zwick
(AYZ)~\cite{AYZ97}. However, our dynamic algorithm cannot directly be
used for the case of $k = 3$ (and only applies for cases $k > 3$)
due to the following challenge which we resolve in
Section~\ref{sec:alg-overview}. Furthermore, our analysis also assumes $k > 6$
for greater simplicity and
since for smaller $k$, our algorithm
from Section~\ref{sec:arboricityclique} is also faster.

\myparagraph{Adapting the Static Algorithm} We face a major challenge when
adapting the algorithm of Alon, Yuster, and Zwick~\cite{AYZ97} for our setting as
well as for the sequential setting. Because the AYZ algorithm is meant to count cliques
in the static setting, it is fine to consider two different types of triangles
and count the triangles of each type separately. The two different types of triangles
considered are triangles which contain at least one low-degree vertex and triangles
which contain only high-degree vertices. In the static case, we can find all
low-degree vertices, but in the dynamic case, we cannot afford to look at all
low-degree vertices. If we only look at low-degree vertices incident to edge
updates, then the following case may occur: an edge update between two high-degree
nodes forms a new triangle incident to a low-degree node. In such a case,
only looking at the vertices adjacent to this edge update will not find this triangle.
We resolve this issue for $k > 3$ via Lemma~\ref{lem:one-low-high} in Section~\ref{sec:alg-overview}.

\myparagraph{Definitions and Data Structures}
Given a graph $G$, we construct an auxiliary graph $G'$ consisting of
vertices where each vertex represents a clique of size $\ell = k/3$ in
$G$.\footnote{We use a hash table $\mathcal{Q}$ that stores each
  vertex in $G'$ as an index to a set of vertices in $G$ and also
  stores each set of vertices composing an $\ell$-clique in $G$
  (lexicographically sort the vertices and turn into a string) as an
  index to a vertex in $G'$.}
An edge $(u, v)$ between two vertices in $G'$ exists if and only if the cliques
represented by $u$ and $v$ form a clique of size $2\ell$ in $G$. Our
algorithm maintains a dynamic total triangle count $C$ on $G'$.
Let $M=2m+1$ and let a \defn{low-degree} vertex in $G'$ be a vertex
with degree less than $M^{t \ell}/2$ (for some $0<t<1$ to be determined
later) and a \defn{high-degree} vertex in $G'$ be a vertex with degree
greater than $3M^{t\ell}/2$. The vertices with degree in the range
$[M^{t\ell}/2, 3M^{t\ell}/2]$ can be classified as either low-degree
or high-degree. In addition to the total triangle count, we maintain a
count, $C_{\low}$, of all triangles involving a low-degree vertex.
Using the algorithm of AYZ~\cite{AYZ97}, we assume we have a two-level hash table, $\low$,
representing the neighbors of low-degree vertices in $G'$ (a table
mapping a low-degree vertex to another hash table containing its
incident edges).  We also maintain the adjacency matrix $A$ of
high-degree vertices in $G'$ used in AYZ as a two-level hash table for
easy insertion and deletion of additional high-degree vertices.
Finally, we maintain another hash table $\degarray$ which dynamically
maintains the degrees of the vertices.

An simplified version of the algorithm is given in Algorithm~\ref{alg:mmcliquesimple}.

\begin{mdframedalg}{Simplified matrix multiplication $k$-clique counting algorithm.}\label{alg:mmcliquesimple}
    \begin{algorithmic}[1]
        \Function{Count-Cliques}{$\batchset$}
            \State Update graph $G'$ with $\batchset$ by inserting new $\ell$- and $2\ell$-cliques.
            \State Find batch of insertions into $G'$, $\batchset'_I$, and batch of deletions, $\batchset'_D$.
            \State Determine the final degrees of every vertex in $G'$ after performing updates $\batchset'_I$
            and $\batchset'_D$.
            \ParFor{$\ins(u, v) \in \batchset'_I, \del(u, v) \in \batchset'_D$}\footnotemark
            \If{either $u$ or $v$ is low-degree: $d(u) \leq \delta$ or $d(v) \leq \delta$}
            \State Enumerate all triangles containing $(u, v)$. Let this set be $T$.
            \State By Lemma~\ref{lem:one-low-high}, find all possible triangles representing the same triangle
            $t \in T$.
            \State Correct for duplicate counting of triangles.
            \Else
            \State Update $A$ (adjacency list for high-degree vertices).
            \EndIf
            \EndParFor
            \State Compute $A^3$. The diagonal provides the triangle counts for all triangles containing only high-degree
            vertices.
            \State Sum the counts of all triangles.
            \State Correct for duplicate counting of cliques.
        \EndFunction
    \end{algorithmic}
\end{mdframedalg}
    \footnotetext{Some care must be taken
                to ensure that rebalancing does not incur too much work. The details of how to deal with
            rebalancing are given in the full implementation, Algorithm~\ref{alg:mmclique}.}

\ifCameraReady
\subsection{Overview}\label{sec:alg-overview}~
\fi
\ifFull
\subsection{Overview}\label{sec:alg-overview}
\fi

Our algorithm proceeds as follows.  Each edge in an update in the
batch (edges in $G$) can either create at most $O(m^{k/3 - 1})$ new
$(2k/3)$-cliques or disrupt $O(m^{k/3 - 1})$ existing $(2k/3)$-cliques
in $G$.
We treat each of these newly created or destroyed cliques as an edge
insertion or deletion in $G'$.
Since we preprocess the updates to $G$ such that there are no duplicate or nullifying updates, a destroyed clique
cannot be created again or vice versa. This means that the set of updates to $G'$ will also contain no
nullifying updates.

Importantly, the AYZ algorithm does not take into account
edge insertions and deletions between two high-degree vertices that create or destroy
triangles containing at least one low-degree vertex.\footnote{Note that this is fine for the static case but not for the dynamic case.}
Thus, we must prove the following lemma for any edge insertion/deletion in $G$
that results in an edge insertion in $G'$ between two high-degree vertices
which creates or destroys a triangle containing a low-degree vertex.
This lemma is crucial for our algorithm, since it ensures that a
triangle formed by two high-degree vertices and a low-degree vertex
will be discovered by enumerating all triangles formed or deleted by an edge update
incident to the low-degree vertex, and its current edges. Furthermore,
this lemma is the reason why our algorithm does not work for $k = 3$ cliques.

\begin{lemma}\label{lem:one-low-high}
Given a graph $G=(V, E)$, the corresponding $G' = (V', E')$, and for $k > 3$,
suppose an edge insertion (resp.\ deletion) between two high-degree vertices in
$G'$ creates a new triangle, $(u_H, w_H, x_L)$, in $G'$ which contains a low-degree vertex $x_L$.
Let $R(y)$ denote the set of vertices in $V$ represented by a vertex $y \in V'$.
Then, there exists a new edge insertion (resp.\ deletion) in $G'$ that
is incident to $x_L$ and creates a new triangle $(u', w', x_L)$
such that $R(u') \cup R(w')  = R(u_H) \cup R(w_H)$.
\end{lemma}

\begin{proof}
We prove this lemma for edge insertions in $G$. The proof
can be easily modified to account for the case of edge deletions in
$G$. Suppose an edge insertion $(y, z)$ in $G$ leads to an edge
insertion in $G'$ between the two high-degree vertices $u_H$ and $w_H$
that creates the new triangle $(u_H, w_H, x_L)$. The creation of the
new triangle signifies that a new clique was created in $G$ consisting
of vertices $R(u_H) \cup R(w_H) \cup R(x_L)$. Then, the edge insertion
$(y, z)$ created a new $2k/3$-clique in $G$ consisting of the vertices
in $R(u_H) \cup R(w_H)$.  Since the edge $(y, z)$ between $y, z \in V$
did not exist previously but now exists, ${2k/3 - 2 \choose k/3 - 2}$
new cliques were created using the set of vertices in $R(u_H) \cup
R(w_H)$.  Each of these new cliques corresponds to a new vertex in
$G'$.  Suppose $u'$ is one such new vertex representing vertex set
$R(u') \subseteq R(u_H) \cup R(w_H)$ and $w'$ represents vertex set
$R(w') = \left(R(u_H) \cup R(w_H)\right) \setminus R(u')$.  Then, new
edges are inserted between $u'$ and $w'$ and between $u'$ and $x_L$
(the edge $(w', x_L)$ might be a newly inserted edge or
it is already present in the graph)
since all triangles representing the clique of vertices $(u_H, w_H,
x_L)$ must be present in $G'$.  Thus, the new triangle $(u', w', x_L)$
is created in $G'$.
\end{proof}

We now describe our dynamic clique counting algorithm that combines
the AYZ algorithm~\cite{AYZ97} with the clique counting algorithm of~\cite{EG04}.
Given the batch of edge insertions/deletions into $G$, we first
compute the duplicate and nullifying updates and remove them.  Then,
for a set of insertions/deletions into $G'$, we form two batches, one
containing the edge insertions and one containing the edge deletions.
Given the batch of updates to $G'$, we now formulate a dynamic version
of the AYZ algorithm~\cite{AYZ97} on the updates to $G'$. For the
batch of updates, we first look at the updates pertaining to the
low-degree vertices.  For every update $(u, v)$ that contains at least
one low-degree vertex (without loss of generality, let $v$ be a
low-degree vertex), we search all of $v$'s $O\left(3M^{t\ell}/2\right)$ neighbors and
check whether a triangle is formed (resp.\ deleted). For each triangle
formed (resp.\ deleted), we update the total triangle count of the
graph $G'$. For high-degree vertices, we update our adjacency matrix
$A$ containing vertices with high-degree. To compute the triangles
containing high-degree vertices, we need only compute $A^3$ (the diagonal will
then provide us with the triangle counts).  Lastly,
one clique results in many different copies of triangles. We must
obtain the correct clique count by dividing the number of triangles by
the number of ways we can partition  the vertices in a $k$-clique into triples of
subcliques of size $k/3$.

\subsection{Detailed Parallel Batch-Dynamic Matrix Multiplication Based Algorithm}\label{sec:matrix-full-impl}
The analysis we perform in Section~\ref{sec:matrix-analysis} on the efficiency of
our algorithm is with respect to the detailed implementation.
We provide the detailed description and implementation of our
algorithm below in Algorithm~\ref{fig:detailed-matrix}.

\begin{breakablealgorithm}{}\label{alg:mmclique}
\caption{Detailed matrix multiplication based parallel batch-dynamic $k$-clique counting algorithm.}\label{fig:detailed-matrix}
\begin{enumerate}[label=(\textbf{\arabic*}),topsep=0pt,itemsep=0pt,parsep=0pt,leftmargin=20pt]
  \item Given a batch $\batchset$ of non-nullifying edge
    updates,\footnote{Recall that we can always remove nullifying edge
      updates as given in Section~\ref{sec:update-alg}.} first update
    the graph $G'$. If the update is an insertion, $\ins(u, v)$, add all
    new $\ell$-cliques created by it into $G'$. If the update is a
    deletion, $\del(u, v)$, mark all $\ell$-cliques destroyed by it
    in $G'$.\footnote{We check in our hash table $\mathcal{Q}$
      whether each newly created (deleted) $\ell$-clique is
      already represented (non-existent) in the graph $G'$. If not, we insert the new
      clique and/or remove an old clique from $\mathcal{Q}$.}
    For each update, $\ins(u, v)$ or $\del(u, v)$, determine all $2\ell$-cliques that
    include it. This will determine the set of edge insertions/deletions into $G'$. Let all edge updates
    that destroy $2\ell$-cliques be a batch $\batchset'_{D}$ of edge
    deletions in $G'$. Then, let all $2\ell$-cliques formed by edge updates be a batch of edge insertions
    $\batchset'_{I}$ into $G'$. Note that edge insertions in the batch could
    be edges for newly created vertices; for each such newly created vertex, we
    also add the vertex into $G'$ and its associated data structures.
    \label{matrix:determineedgeinsertions}

    \item Determine the final degree of each vertex after all insertions in $\batchset'_{I}$
    and all deletions in $\batchset'_{D}$. (We do not perform the updates yet--only compute the final degrees.)
    For all vertices, $X$, which become low-degree after the set of all updates (and were originally high-degree),
    we create a batch of updates $\batchset'_{I, L}$ consisting of old edges (not update edges) that are adjacent
    to vertices in $X$ and were not deleted by the batches of updates. For all vertices, $Y$, which become
    high-degree after the set of updates (and were originally low-degree), we create a batch of updates $\batchset'_{D, H}$
    consisting of old edges adjacent to vertices in $Y$ that were not deleted after the batches of updates.
    \footnote{The batch of updates $\batchset'_{I, L}$ is used to rebalance the data structures when vertices
    need to be removed from $A$ after becoming low-degree. Because the edges adjacent
    to these vertices need to be inserted into the structures maintaining low-degree vertices,
    $\batchset'_{I, L}$, then, can be thought of
    as a set of edge insertions to update low-degree data structures.
    Similarly, vertices which become high-degree need to be
    deleted from low-degree structures, and hence, $\batchset'_{D, H}$ can be thought of
as a set of edge deletions from low-degree structures.} \label{matrix:determinefinaldegree}
    \item Let the edges in $\batchset'_{D} \cup \batchset'_{D, H}$ be the batch of edge deletions to $G'$.
    For each of the edges in $\batchset'_{D} \cup \batchset'_{D, H}$,
      we first count the number of triangles it
      is a part of that contain at least one low-degree vertex. We call this the set of deleted triangles.
      Let this number of deleted triangles be $T_D$ (initially set $T_D = 0$).\label{matrix:computetd}

    \begin{enumerate}
            \item To count the number of triangles that contain at least one low-degree vertex, we first check for each edge whether one of its endpoints is
            low-degree. Let this set of edge deletions be $D'_L \subseteq \batchset'_{D} \cup \batchset'_{D, H}$.
            \item For every edge $(u', v') \in D'_L$,
              without loss of generality let $u'$ be the lexicographically\footnote{The specific lexicographical
            order for the vertices in $G'$ is fixed but
            can be arbitrary.} first low-degree vertex.
          For every edge $(u', w')$ incident to $u'$, check whether
          $(u', v')$ forms a triangle with $(u', w')$.

            \item For every $(u', v', w')$ triangle deleted (where $(u', v', w')$ is sorted lexicographically), call\\
             $t \gets \counttriangles{(u', v', w'), (u', v')}$, and atomically update $T_D \gets T_D + t$.
    \end{enumerate}

    \item Update $C_{\low} \leftarrow C_{\low} -
      T_D$.\label{matrix:updatelowcountone}

    \item Update the data structures using the batches of edges insertions and deletions, $\batchset'_D$ and $\batchset'_I$:\label{matrix:updatestructs}
    	\begin{enumerate}
    	\item Using $\batchset'_{D}$, delete the relevant edges in $\low$ (containing neighbors of low-degree vertices)
     	 and then change the relevant values in $A$ to $0$. We also
      	update $\degarray$ with the new degrees of the vertices for
      	which an adjacent edge was deleted.
      	\label{matrix:updatedegreeone}

      	\item For the batch of edge insertions into $G'$,
      	$\batchset'_{I}$, we first insert the relevant edges into
      	$\low$. Then, we change the relevant entries in $A$ from $0$ to
      	$1$. Finally, we update $\degarray$ with the new degrees of the
      	vertices following the edge insertions.
      	\label{matrix:updateinsertions}

    	\item Remove all vertices which are no longer high-degree
    	(i.e.\ their degree is now less than $M^{t\ell}/2$) from
    	$A$. Create entries in $\low$ for all edges adjacent to each
    	vertex that was removed from $A$.
    	\label{matrix:rebalancelowtohigh}

    	\item Remove the edges of all vertices which are no longer
      low-degree (i.e.\ their degree is now greater than
      $3M^{t\ell}/2$) from $\low$ and create new entries in $A$ with
      the new high-degree vertices. Set the relevant entries in $A$
      corresponding to edges adjacent to the new high-degree vertices
      to $1$.\label{matrix:rebalancehigh}
	\end{enumerate}

      \item Let the edges in $\batchset'_{I} \cup \batchset'_{I, L}$ be the batch of edge insertions to $G'$. For each of the edges in $\batchset'_{I} \cup \batchset'_{I, L}$,
        we first count the
        number of triangles it is a part of that contain at least one low-degree vertex. We call this the set of inserted triangles. Let this value be $T_I$ ($T_I = 0$ initially).
        \label{matrix:countinsertions}

		\begin{enumerate}
			\item To count the number of triangles that contain at least one low-degree vertex, we first check for each edge whether one of its endpoints is
			low-degree. Let this set of edge insertions be $I'_L \subseteq \batchset'_{I} \cup \batchset'_{I, L}$.
			\item For every edge $(u', v') \in I'_L$,
                          without loss of generality let $u'$ be the lexicographically first low-degree vertex.
      For every edge $(u', w')$ of $u'$, check whether $(u', v')$
        forms a triangle with $(u', w')$.
			\item For every newly inserted triangle $(u', v', w')$ (where $(u', v', w')$ is sorted lexicographically), call\\
			 $t = \counttriangles{(u', v', w'), (u', v')}$, and atomically update $T_I \leftarrow T_I + t$.
		\end{enumerate}
    \item Update $C_{\low} \leftarrow C_{\low} +
      T_I$.\label{matrix:updatelowcounttwo}

    \item We perform parallel matrix multiplication after all entries
      in $A$ have been modified to calculate $S = A^3$. Then, $C_{\high} = \frac{1}{2}\sum_{i \in n} S_{i, i}$.
\label{matrix:computeacubed}

    \item Update $C \leftarrow C_{\low} + C_{\high}$.
\label{matrix:updatecount}
    \item Compute the number of $k$-cliques by dividing $C$ by ${k \choose k/3}{2k/3 \choose k/3}$.
\label{matrix:computetriangles}

    \item If $m$ falls outside the range $[M/4,M]$, then reinitialize
      the degree thresholds and data
      structures.\label{matrix:reinitialize}
\end{enumerate}
\end{breakablealgorithm}

Algorithm~\ref{alg:mmclique} uses a subroutine defined below in Algorithm~\ref{alg:subroutine}.

\begin{mdframedalg}{Subroutine used in our detailed matrix multiplication $k$-clique counting algorithm
that counts the number of unique triangles containing an edge.}
\label{alg:subroutine}
\begin{enumerate}[label=(\textbf{\arabic*}),topsep=0pt,itemsep=0pt,parsep=0pt,leftmargin=20pt]
  \item Let $u', v', w' \in V'$ represent the sets of vertices $U',
    X', W' \subseteq V$, respectively.

  \item Enumerate all possible triangles that represent the clique
    containing vertices $U' \cup X' \cup W'$.\label{count:enum}

   \item Sort the vertices of each triangle lexicographically to
     obtain tuples of vertices representing the triangles.  Let
     $\id(u', v')$ be the ID of edge $(u', v')$.\footnote{There are
       many possible ways to assign IDs to edges--for example, the ID
       of an edge could be the concatenation of the IDs of the
       vertices composing the edge.}

  \item For each enumerated tuple $(x', y', z')$, create a label
    containing the tuple representing the triangle concatenated with
    all labels (sorted lexicographically) of edges that are updates in
    the triangle.  Thus, each label can have $4$ to $6$ entries
    consisting of the three vertices of a triangle tuple and at most $3$
    edge labels.  For example, suppose that $(x', y')$ is the only
    edge that is an updated edge in triangle $(x', y', z')$. Then, the
    label representing this triangle is $(x', y', z', \id(x', y'))$
    where the ID of the edge is given by $\id(x', y')$. The IDs of all
    deleted or inserted edges are appended to the end of the label in
    the order $\id(x', y'), \id(y', z'), \id(z', x')$.

  \item Sort all labels lexicographically.

  \item Without loss of generality, let $L = (x', y', z', \id(x',
    y'))$ be the lexicographically-first of these triangle labels
    which contains at least one edge deletion (resp.\ edge insertion)
    of an edge that is incident to at least one low-degree vertex.

  \item If $(u', v', w')$ corresponds to the lexicographically-first
    label $L$ \emph{and} $\id(u', v')$ is the first edge ID in the
    label that contains a low-degree vertex, then $(u', v')$ performs
    the following steps:

    \begin{enumerate}
            \item Count the number of unique triangles (using the
              labels, one can count the unique triangles) containing
              at least one edge deletion (resp.\ insertion) and at
              least one low-degree vertex as $T_D$ (resp.\ $T_I$). We
              count using the generated labels for the triangles
              enumerated in step~\ref{count:enum} of this procedure.
              \item Return $T_D$ (resp.\ $T_I$).
    \end{enumerate}

  \item If $(u', v', w')$ is not equal to $L$ \emph{or} $\id(u', v')$
  is not the first edge ID that contains a low-degree vertex in the label, return $0$.
\end{enumerate}
\end{mdframedalg}

\ifCameraReady
\subsection{Analysis}\label{sec:matrix-analysis}~
\fi
\ifFull
\subsection{Analysis}\label{sec:matrix-analysis}
\fi

In Theorem~\ref{lem:correctness-matrix}, we prove that the procedure
correctly returns the exact number of $k$-cliques in $G$. The proof is
similar to AYZ except that each $\ell$-clique can appear multiple
times in $G'$ so we need to normalize by the constant stated in step~\ref{matrix:computetriangles} of Algorithm~\ref{alg:mmclique}.

\begin{restatable}{theorem}{matrixcorrectness}\label{lem:correctness-matrix}
Algorithm~\ref{alg:mmclique} correctly computes the exact number of
cliques in a graph $G = (V, E)$ when $k \bmod 3 = 0$.
\end{restatable}

\begin{proof}
We first show that all triangles in $G'$ represent a $k$-clique in
$G$. A vertex exists in $G'$ if and only if it is a $(k/3)$-clique in
$G$. Similarly, an edge exists in $G'$ if and only if it connects two vertices in
$G'$ that form a $(2k/3)$-clique in $G$. Thus, a triangle connects $3$
pairs of $3$ distinct $(k/3)$-cliques. This implies that each pair
represents a complete subgraph, which necessarily means by the
pigeonhole principle that the triangle represents a $k$-clique. Now we
show that for each unique $k$-clique in $G$, there exist exactly ${k
  \choose k/3}{2k/3 \choose k/3}$ triangles representing it in $G'$. For each
$k$-clique in $G$, there are ${k \choose k/3}$ distinct
$(k/3)$-subcliques. Each of these subcliques is represented by a
vertex in $G'$. Each distinct triple of subcliques will be a triangle
in $G'$. There are ${k \choose k/3}$ ways to choose the first
subclique, ${2k/3 \choose k/3}$ ways to choose the second subclique,
and ${k/3 \choose k/3}$ ways to choose the third subclique in the
triple. Thus, the total number of duplicate triangles is ${k \choose
  k/3}{2k/3 \choose k/3}$.

We conclude by proving that our algorithm finds the exact number of
triangles in $G'$. All triangles containing edge updates where at
least one of its endpoints is low-degree can be found by searching all
of the neighbors of the low-degree vertex.
All such neighbors will be in $\low$, thus,
searching through the entries in $\low$ is enough to
find all triangles containing at least one low-degree vertex and an edge update
to a low-degree vertex.
By Lemma~\ref{lem:one-low-high}, all triangles with
a low-degree vertex, containing a single edge update
between high-degree vertices can be found via the $\mathtt{count\_new\_low\_degree\_triangles}$
procedure. The same logic handles vertices that change status from
high-degree to low-degree, since we treat edges incident to these
vertices as new edge insertions.
Finally, the procedure ensures that
no duplicate triangles are added to the update triangle count
because the lexicographically first triangle
counts all possible triangles representing the same clique (and no others increment the count).  Table
$A$ is used to compute (via transitive closure) the number of
triangles that contain no low-degree vertices. Thus, by computing
$A^3$, we find the remaining triangles which only contain high-degree
vertices.
Finally, dividing by the total number of different triangles that are created per
unique clique gives us the precise count of the number of $k$-cliques in $G$.
\end{proof}

\myparagraph{Cost}
We now analyze the work, depth, and space of the dynamic algorithm.
Our analysis assumes that $m^{\op/(1+\op)} = O(m^{t\ell})$ so that the $O(m^{t\ell})$ terms in our analysis are only affected by a constant factor for our batch size of $\batch \leq m^{\op/(1+\op)}$.
This is true for $k > 6$ because $t\geq 1/3$ and $\ell \geq 3\op/(1+\op)$. For small $\ell$ we use the combinatorial
algorithm from Section~\ref{sec:arboricityclique}, which is also faster.

First, we compute the work and depth bound of performing preprocessing
on an initial graph $G = (V, E)$ with $m$ edges. We can also apply this preprocessing directly without running the update algorithm
whenever we receive a batch of size $\batch > m^{\op/(1+\op)}$.

For preprocessing, we use a different threshold $m^{t'\ell}$ for low-degree and high-degree vertices.
Searching for all the triangles containing at least one low-degree vertex takes $O\left(m^{(1+t')\ell}\right)$ work by a similar calculation as in Lemma~\ref{lem:compute-low-deg} and searching for triangles containing all high-degree vertices takes $O\left(m^{(1-t')\ell\op}\right)$ work by Lemma~\ref{lem:compute-from-scratch}. Thus, the optimal value $t'$ is when $m^{(1+t')\ell} = m^{(1-t')\ell\op}$, which gives $t'=\frac{\op -1}{\op + 1}$ as in~\cite{AYZ97}.

\begin{restatable}{lemma}{preprocessing}\label{lem:mmpreprocessing}
Preprocessing the graph $G = (V, E)$ with $m$ edges into $G'$,
creating the data structures $\low$, $A$, and $\degarray$, and
counting the number of $k$-cliques takes $O\left(
m^{\frac{2k\op}{3(1+\op)}}\right)$ work and $O(\log m)$ depth w.h.p.,
and $O\left(m^{\frac{2k\op}{3(1+\op)}}\right)$
space assuming a parallel matrix multiplication algorithm with
coefficient $\op$. Using the fastest parallel matrix multiplication currently known (\cite{LeGall14}, Corollary~\ref{cor:matrix-exp}),
preprocessing takes $O\left(m^{0.469k}\right)$ work and $O(\log m)$ depth w.h.p., and
$O(m^{0.469k})$ space.
\end{restatable}

\ifFull
\begin{proof}
The graph $G'$ has size $O(m^{\ell})$ by
Lemma~\ref{lem:edge-clique-bound}. We can find all $\ell$-cliques using
$O(m^{\ell/2})$ work and $O(1)$ depth and all $2\ell$-cliques using $O(m^{\ell})$ work and $O(1)$ depth.
Initializing the data structures
$\low$ and $\degarray$ with $O(m^{\ell})$ entries requires
insertions into two parallel hash tables. This takes $O(m^{\ell})$ work and $O(\log^* m)$ depth w.h.p., and
$O(m^{\ell})$ space. There are $O\left(m^{\frac{2\ell}{(1+\op)}}\right)$
high-degree vertices which means that initializing $A$, the adjacency
matrix, requires creating a $2$-level hash table with
$O\left(m^{\frac{4\ell}{(1+\op)}}\right)$ entries. This takes
$O\left(m^{\frac{4\ell}{(1+\op)}}\right)$ work and $O(\log^* m)$ depth w.h.p., and
$O\left(m^{\frac{4\ell}{(1+\op)}}\right)$ space. Computing $A^3$
requires $O\left(m^{\frac{2\ell\op}{(1+\op)}}\right)$ work, $O(\log m)$
depth, and
$O\left(m^{\frac{2\ell\op}{(1+\op)}}\right)$
space. Finally, counting all the triangles with at least one
low-degree vertex requires $O\left(m^{\frac{2\ell\op}{(1+\op)}}\right)$
work and $O(1)$ depth (by performing $O\left(m^{(1+t)\ell}\right)$
lookups in $\low$). By Corollary~\ref{cor:matrix-exp}, $\op = 2.373$, and since $\ell = k/3$,
preprocessing takes $O\left(m^{0.469k}\right)$ work, $O(\log m)$
depth, and $O(m^{0.469k})$ space.
\end{proof}
\fi

Next, we analyze the update procedure of our dynamic algorithm.
To start, we bound the number of vertices and edges in $G'$ (representing the number
of $\ell$ and $2\ell$ cliques in $G$, respectively) in terms of $m$
(the number of edges in $G$) below.

\begin{lemma}[\cite{Chiba1985}]\label{lem:edge-clique-bound}
Given a graph $G = (V, E)$ with $m$ edges, the number of $k$-cliques that $G$ can have is bounded by $O(m^{k/2})$.
\end{lemma}

\begin{lemma}\label{lem:space}
$G'$ uses $O(m^{\ell})$ space.
\end{lemma}

\begin{proof}
  Each vertex in $G'$ represents an $\ell$-clique.
By Lemma~\ref{lem:edge-clique-bound}, $G'$ has $O(m^{\ell/2})$ vertices and thus $O(m^{\ell})$ edges.
\end{proof}

Before we compute the number of triangles in $G'$, we must update $G'$ and the data structures associated with $G'$ with our batch of updates.

\begin{restatable}{lemma}{updatingstructs}\label{lem:updating-structs}
Updating $G'$ and the associated data structures $\low$ and $A$ after a batch of $\batch$ edge updates in $G$ takes $O(\batch m^{\ell - 1} + \batch m^{(2-2t)\ell - 1})$ amortized work and $O(\log^* m)$ depth w.h.p., and $O\left(m^{\ell} + m^{(2-2t)\ell}\right)$ space.
\end{restatable}

\ifFull
\begin{proof}
In step~\ref{matrix:determineedgeinsertions} we first add and/or delete vertices in $G'$. Since each vertex in $G'$
represents a different clique of size $\ell$, one edge update in $G$
can result in $O(m^{(\ell/2) - 1})$ new vertices (or vertex deletions)
since given two vertices (the endpoints of the edge update) that must
be in the $\ell$-clique, we only need to look for all $(\ell - 2)$-cliques in $G$. For a batch of size $\batch$, the total number of
vertices added or deleted in $G'$ is $O(\batch m^{(\ell/2) -1})$.

In steps~\ref{matrix:updatedegreeone} and~\ref{matrix:updateinsertions},
updating the data structures $\low$, $A$, and $\degarray$ by
insertions/deletions into parallel hash tables requires $O(\batch m^{\ell-1})$
amortized work and $O(\log^*m)$ depth w.h.p.
Recall that
the number of edges in $G'$ is determined by the total number of
$2\ell$-cliques in $G$. One edge update can affect at most $O(m^{\ell
  - 1})$ $2\ell$-cliques in $G$, thus, given a $\batch$-batch of edge updates
in $G$, there will be $O(\batch m^{\ell - 1})$ edge updates in $G'$,
separated into a deletion batch $\batchset'_D$ and an insertion batch
$\batchset'_I$.

We now analyze the cost for steps~\ref{matrix:rebalancelowtohigh} and~\ref{matrix:rebalancehigh}.
Adding/removing a row and column from $A$ takes $O(m^{(1-t)\ell})$ amortized work. Since there are $O(m^{\ell-1})$ edge updates in $G'$ per update in $G$, the total work for resizing is $O(m^{(2-t)\ell-1})$ per edge update in $G$. The work for adding/removing a vertex from $\low$ is $O(m^{t\ell})$, and since there are $O(m^{\ell-1})$ edge updates per update in $G$, the total work is $O(m^{(1+t)\ell-1})$ per update in $G$.
We must have $\Omega(m^{t\ell})$ updates in $G'$ before a vertex changes statuses (becomes high-degree if it originally was low-degree and vice versa) and needs to update $A$ and $\low$.
Therefore, we can charge the work of updating $A$ and $\low$ against $\Omega(m^{t\ell})$ updates in $G'$.
Thus, the amortized work for updating $A$ and $\low$ given a batch of $\batch$ updates in $G$ is $O\left(\batch\left(m^{(2-2t)\ell-1}+m^{\ell-1}\right)\right)$ for steps~\ref{matrix:determineedgeinsertions} and~\ref{matrix:updatestructs}. The depth is $O(\log^*m)$ w.h.p.\ due to hash table operations.

The data structures $\low$, $\degarray$, and $A$ use a combined
$O(m^{\ell} + m^{(2-2t)\ell})$ space because there are $O(m^\ell)$
edges in the graph and $A$ contains $O(m^{(2-2t)\ell})$ entries.
\end{proof}
\fi

By Lemma~\ref{lem:updating-structs}, step~\ref{matrix:determinefinaldegree} takes $O\left(\batch m^{\ell-1}\right)$ amortized work to determine the final
degrees and  $O(\batch m^{\ell - 1} + \batch m^{(2-2t)\ell - 1})$  amortized work to compute $B'_{I, L}$ and $B'_{D, H}$. In total, step~\ref{matrix:determinefinaldegree}
takes $O(\batch m^{\ell - 1} + \batch m^{(2-2t)\ell - 1})$ amortized work, $O(\log m)$ depth (dominated by computing the final degrees), and $O(m^{\ell} + m^{(2-2t)\ell})$ space by Lemma~\ref{lem:updating-structs}.
Steps \ref{matrix:updatelowcountone},
\ref{matrix:updatelowcounttwo}, \ref{matrix:updatecount}, and
\ref{matrix:computetriangles} of the algorithm take $O(1)$ work. The
following lemmas bound the cost for the remaining steps.

Lemma~\ref{lem:compute-low-deg} below bounds the cost for steps
\ref{matrix:computetd} and \ref{matrix:countinsertions}. The proof is
based on counting the number of new edge updates necessary in $G'$.

\begin{restatable}{lemma}{mmlowdeg}\label{lem:compute-low-deg}
Computing all new $k$-cliques represented by triangles that contain at least one low-degree vertex in $G'$
takes $O(\batch m^{(t + 1)\ell - 1})$ work and $O(\log^* m)$ depth w.h.p., and $O(m^{\ell})$ space.
\end{restatable}

\ifFull
\begin{proof}
We first bound the work necessary to perform
steps~\ref{matrix:computetd} and~\ref{matrix:countinsertions} for new
edge insertions and deletions.  Given one edge update in $G$, there
can be at most $O(m^{\ell - 1})$ edge updates necessary in $G'$ by
Lemma~\ref{lem:edge-clique-bound}. For each of these edge updates, we
consider whether each edge update in $G'$ contains a low-degree
vertex. By Lemmas~\ref{lem:one-low-high}
and~\ref{lem:correctness-matrix}, to find all updated triangles
containing at least one low-degree vertex, it is only necessary to
consider edge updates to low-degree vertices.  For every edge update
to a low-degree vertex, we search the neighbors of that low-degree
vertex to see if new triangles are formed/destroyed. Since each
low-degree vertex has degree $O(m^{t\ell})$, this results in a total
of $O(m^{(t + 1)\ell - 1})$ work per update in $G$ to perform the
search.  For each triangle found that contains the low-degree vertex,
we need to perform the additional work of computing every triangle
that contains the set of vertices represented by the triangle, sort
the labels, and determine which triangle is responsible for
incrementing the count of triangles by all ${k \choose k/3}{2k/3
  \choose k/3}$ triangles representing the same clique. This
additional work is done by calling $\counttriangles{(u', v', w'), (u',
  v')}$ on each triangle $(u', v', w')$ and each edge update $(u',
v')$.  The total amount of additional work done for each triangle that
is passed into $\mathtt{count\_updated\_low\_degree\_triangles}$ is
then $O\left(k(3e^2)^k\right)$, where the number of triangles
corresponding to the same $k$-clique is given by
$O\left((3e^2)^k\right)$ and an additional $O(k(3e^2)^k)$ work is
required to sort all the labels. Since we assume that $k$ is constant,
this results in $O(1)$ additional work per call to
$\mathtt{count\_updated\_low\_degree\_triangles}$.  The depth is
$O(\log^*m)$ w.h.p.\ due to hash table lookups.

Now we bound the work of performing steps~\ref{matrix:computetd} and~\ref{matrix:countinsertions} for
edges that are `inserted' or `deleted' due to rebalancing. Suppose there are $X$ vertices
that must be rebalanced in this way. Each of these $X$ vertices must have degree $O(m^{t\ell})$
at the time of rebalancing. Thus, the total work performed for these updates is
$O(Xm^{2t\ell})$. However, in order for a rebalancing on a vertex to happen, there must be
$\Omega(m^{t\ell})$ updates. Thus, if $X$ vertices are rebalanced, then there must be $\Omega(Xm^{t\ell})$ updates.
Hence, we can charge the work of rebalancing to the $\Omega(Xm^{t\ell})$ updates to obtain $O(m^{t\ell})$ amortized work
per update in $G'$. Then, we obtain $O(\batch m^{(t + 1)\ell - 1})$ amortized work for a $\batch$ batch updates to $G$.
Rebalancing requires $O(\log^*m)$ depth w.h.p.\ due to hash table operations and $O(m^{\ell})$ space (the total number of edges
in the graph).
\end{proof}
\fi

Lemma~\ref{lem:compute-from-scratch} bounds the cost for step
\ref{matrix:computeacubed} by using the matrix multiplication bounds
for the adjacency matrix containing high-degree vertices.

\begin{restatable}{lemma}{mmhighdeg}\label{lem:compute-from-scratch}
Computing $A^3$ using parallel matrix multiplication takes $O(m^{(1-t)\ell \op})$ work, where $\op$ is the parallel matrix multiplication constant, $O(\pmmdepth{m})$ depth, and $O(m^{\op(1-t)\ell})$ space, assuming that there exists a parallel matrix multiplication algorithm with coefficient $\op$ and using $O(\pmmdepth{n})$ depth and $O(\pmmspace{n})$ space given $n \times n$ matrices.
\end{restatable}

\begin{proof}
  There are $O(m^{(1-t)\ell})$ high-degree vertices because each high-degree vertex has degree $\Omega(m^{t\ell})$ and there are $O(m^{\ell})$ edges in $G'$. Since the table $A$ is an adjacency matrix on the high-degree vertices, by Corollary~\ref{cor:matrix-exp}, parallel matrix multiplication can be done in $O(m^{(1-t)\ell\op})$ work.

\end{proof}

Lemma~\ref{lem:major-rebalancing}
bounds the cost for step \ref{matrix:reinitialize}. The proof is based on amortizing the cost for reconstruction over $\Omega(m)$ updates.
\begin{restatable}{lemma}{mmrebalancing}\label{lem:major-rebalancing}
Step~\ref{matrix:reinitialize} requires $O(\Delta m^{(2-2t)\ell-1}+\Delta m^{\ell-1})$ amortized work and $O(\log^*m)$ depth w.h.p., and $O(m^{(2-2t)\ell}+ m^\ell)$ space.
\end{restatable}

\begin{proof}
  We reconstruct $A$ from scratch, which has one entry for every
  pair of high-degree vertices, which takes
  $O(m^{2(1-t)\ell})=O(m^{(2-2t)\ell})$ work and space. However, this is
  amortized against $\Omega(m)$ updates, and so the amortized work is
  $O(m^{(2-2t)\ell-1})$ per update. The work and space for creating $\low$ can
  be bounded by $O(m^{\ell})$, the number of edges in $G'$. Amortized
  against $\Omega(m)$ updates gives $O(m^{\ell-1})$ work per
  update. The depth is $O(\log^*m)$ w.h.p.\ using parallel hash table
  operations.
\end{proof}

Given these costs, we can now compute the optimal value of $t$ in terms of $\op$ that minimizes the work. Note that here we compute for $t$ assuming $\batch = 1$ because to adaptively change our threshold requires too much work in terms of rebalancing the data structures. However, if we have a fixed batch size, $\batch$, we can further optimize our threshold $t$ to take into account the fixed batch size.

\begin{restatable}{lemma}{mmoptimalt}\label{lem:t-value}
$t = \frac{3 - k + k\op}{k + k\op}$ gives us an optimal work bound assuming $\batch = 1$.
\end{restatable}

\begin{proof}
From
Lemmas~\ref{lem:updating-structs},~\ref{lem:compute-low-deg},~\ref{lem:compute-from-scratch}, and~\ref{lem:major-rebalancing}, we have that the work is $O(\batch
m^{(t + 1)\frac{k}{3} - 1} + m^{\frac{(1-t)k\op}{3}})$
w.h.p. (the $O(\batch m^{(2-2t)l-1})$ term is dominated by the $O(\batch m^{(1+t)l-1})$ term since $\op \geq 2$ implies $t\geq 1/3$).  Assuming $\batch = 1$, balancing the two
sides of the equation yields: $$m^{\frac{(1-t)k\op}{3}} = m^{(t +
  1)\frac{k}{3} - 1}.$$ Solving for $t$ gives $$t = \frac{3 - k +
  k\op}{k + k\op}.$$
\end{proof}

Plugging in our value for $t$ from Lemma~\ref{lem:t-value}, we prove Theorem~\ref{thm:mm-main} and Corollary~\ref{cor:strassen-ws} for the cost of our algorithm when $0 < m \leq m^{\op/(1+\op)}$.

\ifCameraReady
\subsection{Accounting for $k \bmod 3 \neq 0$}\label{sec:all-k-alg}~
\fi
\ifFull
\subsection{Accounting for $k \bmod 3 \neq 0$}\label{sec:all-k-alg}
\fi

We now modify the algorithm above to account for all values $k$
following the algorithm presented in~\cite{EG04}. This requires
several changes to how we construct our graph $G'$ from a graph $G =
(V, E)$, resulting in changes to our data structures which we detail
below. We recall the notation $R(x)$ for vertex $x \in G'$ to denote
the vertices in $G$ that $x$ represents.

\subsubsection{Construction of $G'$}\label{sec:gp-const}

For  $k \bmod 3 \neq 0$, the fundamental problem we face in this case in constructing the graph $G'$ is that triangles in the graph $G'$ representing cliques of size $\floor{\frac{k}{3}}$ no longer create $k$-cliques. In fact, they now create $(k-1)$-cliques or $(k-2)$-cliques for $k \bmod 3 =1$ and $k\bmod3 = 2$, respectively. We modify the creation of $G'$ in the two following ways to account for this issue:

\paragraph{$k\bmod3 = 1$:}
In this case, we create two sets of vertices. One set, $A$, of
vertices represents all
$\left(\frac{k-1}{3}\right)$-cliques in the graph $G$.
Edges exist between $v_1, v_2 \in
A$ if and only if the vertices, $R(v_1)$ and $R(v_2)$, in the
$\left(\frac{k-1}{3}\right)$-cliques represented by $v_1$ and $v_2$
form a $\frac{2(k-1)}{3}$ clique and there are no duplicate vertices,
i.e., $R(v_1) \cap R(v_2) = \emptyset$. We create a second set of
vertices $B$ which contains vertices which represent cliques of size
$\frac{k+2}{3}$. Edges exist between $v \in A$
and $w \in B$ if and only if $R(v)$ and $R(w)$ form a $\left(\frac{2k
  +1}{3}\right)$-clique and $R(v) \cap R(w) = \emptyset$.

\paragraph{$k\bmod3 = 2$:}
In this case, we still create two sets of vertices but $A$
instead represents
$\left(\frac{k+1}{3}\right)$-cliques in the graph $G$.
Edges exist between $v_1, v_2 \in A$ if and only if $R(v_1) \cup R(v_2)$ form a
$\left(\frac{2(k+1)}{3}\right)$-clique and $R(v_1) \cap R(v_2) =
\emptyset$. We create a second set of vertices $B$ which contains
vertices which represent cliques of size $\frac{k-2}{3}$. Edges exist
between $v \in A$ and $w \in B$ if and only if
$R(v)$ and $R(w)$ form a $\left(\frac{2k -1}{3}\right)$-clique and
$R(v) \cap R(w) = \emptyset$.

We first prove the properties the new graph $G'$ has, namely the
number of vertices it contains as well as the number of edges in the
graph.

\begin{lemma}\label{lem:gp-struct}
$G'$ constructed as in Section~\ref{sec:gp-const} contains $O\left(m^{\frac{k+2}{6}}\right)$ vertices and $O\left(m^{\frac{2k+1}{6}}\right)$ edges if $k\bmod 3 = 1$. $G'$ contains $O\left(m^{\frac{k+1}{6}}\right)$ vertices and $O\left(m^{\frac{k+1}{3}}\right)$ edges if $k \bmod 3 =2$.
\end{lemma}

\begin{proof}
When $k\bmod 3 = 1$, the number of vertices is upper bounded
(asymptotically) by the number of $\left(\frac{k+2}{3}\right)$-cliques
in the graph. By Lemma~\ref{lem:edge-clique-bound}, the number of
vertices is then bounded by $O\left(m^{\frac{k+2}{6}}\right)$. The
number of edges is bounded by the number of $\left(\frac{2k
  +1}{3}\right)$-cliques in the graph which is
$O\left(m^{\frac{2k+1}{6}}\right)$. Similarly, when $k\bmod 3 =2$, by
Lemma~\ref{lem:edge-clique-bound}, the number of vertices and edges
are bounded by $O\left(m^{\frac{k+1}{6}}\right)$ and
$O\left(m^{\frac{k+1}{3}}\right)$, respectively.
\end{proof}

\subsubsection{Data Structure and Algorithm Changes}

The major data structure change is to redefine the high-degree and
low-degree vertices in terms of the number of edges in the graph. This
means that low-degree is defined as having a degree less than
$\frac{M^{t\left(\frac{2k+1}{6}\right)}}{2}$ and high-degree as
greater than $\frac{3M^{t\left(\frac{2k+1}{6}\right)}}{2}$ for the
$k\bmod 3 = 1$ case; similarly we define low-degree to be less than
$\frac{M^{t\left(\frac{k+1}{3}\right)}}{2}$ and high-degree to be
greater than $\frac{3M^{t\left(\frac{k+1}{3}\right)}}{2}$ for the
$k\bmod 3 = 2$ case.

Another key difference between this case and the case when $k$ is
divisible by $3$ is that the number of duplicate cliques is different for these two cases. For the
$k\bmod 3=1$ case, each $k$-clique in $G$ will be represented by ${k \choose
  (k+2)/3}{(2k-2)/3 \choose (k-1)/3}$ triangles found by the
algorithm. For the $k\bmod 3 = 2$ case, each $k$-clique in $G$ will be
represented by ${k \choose (k-2)/3}{(2k+2)/3 \choose (k+1)/3}$
triangles. Thus, at the
end of our algorithm, we must divide the count of the triangles by
their respective number of duplicates.

The rest of the algorithm remains the same as before, except that we
solve for different values of $t$ depending on the case. Since the
proofs for obtaining the following results are nearly identical to the
ones for $k \bmod 3 = 0$, we do not restate the proofs and only give
our results.

\begin{lemma}\label{lem:t-mod1}
For the case when $k \bmod 3 = 1$, there exists $O\left(m^{\frac{2k+1}{6}}\right)$ edges in the graph and solving for the
optimal value of $t$ (assuming $\batch = 1$) gives $t = \frac{2k\op - 2k + \op + 5}{2k\op + 2k + \op + 1}$. For the case when $k \bmod 3 = 2$,
there exists $O\left(m^{\frac{k+1}{3}}\right)$ edges in the graph and solving for the optimal value of $t$ gives $t = \frac{k\op - k + \op + 2}{k\op + k + \op + 1}$.
\end{lemma}

Using our values for $t$, we can obtain our final theorem, Theorem~\ref{thm:all-k}, for the work and depth bounds for these two cases.

\begin{theorem}\label{thm:all-k}
Our fast matrix multiplication based $k$-clique algorithm takes\\ $O\left(\min\left(\batch m^{\frac{2(k - 1)\op}{3(\op + 1)}}, (\batch+m)^{\frac{(2 k + 1)\op}{3 (\op + 1)}}\right)\right)$ work and $O(\log(m+\batch))$ depth w.h.p., and
$O\left((\batch+m)^{\frac{(2 k + 1) \op}{3 (\op + 1)}}\right)$ space assuming a parallel matrix multiplication algorithm with coefficient $\op$ when $k \bmod 3 =  1$,  and $O\left(\min\left(\batch m^{\frac{(2k - 1)\op}{3(\op + 1)}}, (\batch+m)^{\frac{2(k + 1)\op}{3(\op + 1)}}\right)\right)$ work and $O(\log(m+\batch))$ depth w.h.p., and $O\left((\batch+m)^{\frac{2(k + 1)\op}{3(\op + 1)}}\right)$ space when $k \bmod 3 =  2$.
\end{theorem}

\begin{corollary}\label{cor:strassen-all-k}
Using Corollary~\ref{cor:matrix-exp} with $\op = 2.373$, we obtain a parallel fast matrix multiplication $k$-clique algorithm that takes $O\left(\min\left(\batch m^{0.469k - 0.469}, (\batch+m)^{0.469k + 0.235}\right)\right)$ work and $O(\log m)$ depth w.h.p., and $O\left((\batch+m)^{0.469k + 0.235}\right)$ space when $k \bmod 3 = 1$, and $O\left(\min\left(\batch m^{0.469k - 0.235}, (\batch+m)^{0.469k +
0.469}\right)\right)$ work and $O(\log m)$ depth w.h.p., and $O\left((\batch+m)^{0.469k +
0.469}\right)$ space when $k\bmod 3 = 2$.
\end{corollary}

\subsection{Parallel Fast Matrix Multiplication}\label{sec:pmm}

In this section, we show that tensor-based matrix multiplication
algorithms (including Strassen's algorithm) can be parallelized in
$O(\log n)$ depth and $O(n^{\omega})$ work. Such techniques are used
for algorithms that achieve the best currently known matrix
multiplication exponents~\cite{Williams12,LeGall14}.  We assume, as is
common in models such as the arithmetic circuit model, that field
operations can be performed in constant work. We refer readers
interested in learning more about current techniques in fast matrix
multiplication to~\cite{Blaser13,alman19}.

Before we prove our main parallel result in this section, we first
define the \emph{matrix multiplication tensor} as used in previous
literature.

\begin{definition}[Matrix Multiplication Tensor (see, e.g., \cite{alman19})]\label{def:mm-tensor}
For positive integers $a, b, c$, the matrix multiplication tensor $\langle a, b, c \rangle$
is a tensor over $\left\{x_{ij}\right\}_{i \in [a], j \in [b]}, \left\{y_{jk}\right\}_{j \in [b], k \in [c]},
\left\{z_{ki}\right\}_{k \in [c], i \in [a]}$, where
\begin{align*}
\langle a, b, c \rangle = \sum_{i=1}^a \sum_{j=1}^b \sum_{k=1}^c x_{ij} y_{jk} z_{ki}.
\end{align*}
\end{definition}

The matrix multiplication tensor can be seen as a generating function for $A \times B$ multiplication where
the coefficients of the $z_{ki}$ terms are exactly the $(i, k)$ entries in the matrix product $A \times B$ where
$A = \begin{pmatrix}
x_{11} &\dots & x_{1b}\\
\dots&\dots&\dots \\
x_{a1} & \dots & x_{ab}
\end{pmatrix}$ and $B = \begin{pmatrix}
y_{11} &\dots & y_{1c}\\
\dots&\dots&\dots \\
y_{b1} & \dots & y_{bc}
\end{pmatrix}$.

Current matrix multiplications algorithms use this fact to obtain the best known exponents.
The proof of the following lemma closely follows the proof of Proposition 4.1 given in~\cite{alman19}.

\begin{lemma}\label{lem:matrix-mult-tensor-parallel}
Let $R\left(\langle q, q, q\rangle\right) \leq r$ (over a field $\mathbb{F}$) be the rank of the matrix multiplication tensor $\langle q, q, q\rangle$.
Assuming that field operations take $O(1)$ work, then,
there exists a parallel matrix multiplication algorithm that performs $A \times B$ matrix multiplication (where $A, B \in \mathbb{F}^{n \times n}$)
over $\mathbb{F}$ using $O\left(n^{\log_q(r)}\right)$ work and $O((\log r + \log q)\log_q n)$ depth using $O\left(n^{\log_q(r)}\right)$ space.
\end{lemma}

\begin{proof}
By definition of rank, since $R \left(\langle q, q, q\rangle\right) \leq r$,
\begin{align*}
\langle q, q, q\rangle = \sum_{\ell = 1}^r \left(\sum_{i, j \in [q]}
a_{ij\ell}x_{ij}\right)\left(\sum_{j, k \in [q]}
b_{jk\ell}y_{jk}\right)\left(\sum_{k, i \in [q]} c_{ki\ell}z_{ki}\right)
\end{align*}
for some coefficients $a_{ij\ell}, b_{jk\ell}, c_{ki\ell} \in \mathbb{F}$. Computing this matrix multiplication tensor requires at most $O\left(rq^2\right)$ field operations.

Using this information, we perform parallel matrix multiplication via the following recursive algorithm. We assume that $n$ is a power of $q$; otherwise,
we can pad $A$ and $B$ with $0$'s until such a condition is satisfied--this would increase the dimensions by at most a factor of $q$.

Partition the padded matrices $A$ and $B$ into $q \times q$ block matrices where each block has size $n/q \times n/q$.
This algorithm performs, in parallel, the following linear combinations for each $\ell$,

\begin{align*}
A_{\ell}' = \sum_{i, j \in [q]} a_{ij \ell} A_{ij} \\
B_{\ell}' = \sum_{j, k \in [q]} b_{jk \ell} B_{jk}
\end{align*}
where $A_{ij}$ and $B_{jk}$ are the $n/q \times n/q$ blocks in $A$ and $B$, respectively. Such operations
require $O(rq^2)$ operations to perform; however, all such multiplication operations can be done in parallel, and the
summation of the results can be done in $O(\log q)$ depth, resulting in $O(\log q)$ depth.

Then, for each $\ell \in [r]$, we compute $C_{\ell}' = A_{\ell}' \times B_{\ell}'$ by performing parallel $n/q \times n/q$
matrix multiplication recursively on $A_{\ell}'$ and $B_{\ell}'$ where the base case is $q \times q$ matrix multiplication.
All field operations in the same level of the recursion can be performed in parallel. There are $O(\log_q n)$ levels of recursion.
Each level of recursion computes a number of field operations in parallel in $O(\log q)$ depth as in the top level.

Finally, after obtaining the results $C_{\ell}'$ of the recursive calls, we compute

\begin{align*}
C_{ki} = \sum_{\ell \in [r]} c_{ki\ell} C_{\ell, ki}'
\end{align*}
for all $k, i \in [q]$ where $C_{\ell, ki}'$ are the results we obtain from our recursive calls.
The blocks $C_{ki}$ for all $k, i \in [q]$ are the results of our matrix multiplication $A \times B$.

This final step can compute in parallel the blocks $C_{ki}$ for all $k, i \in [q]$ in $O(\log r)$ depth
(assuming that we have the results $C_{\ell, ki}'$) since the multiplication operations can be done in parallel
and the summation of the elements in the resulting matrices can be done in $O(\log r)$ depth.

Thus, the depth required for this algorithm is $O((\log r + \log q)\log_q n)$.

To compute the work and space usage, we compute the total number of field operations performed,
which is $O(n^2)$ per level of the recursion. For each level of recursion, there are $r$
calls per subproblem of the recursion. Since we assume that each field operation is $O(1)$ work,
this results in total work given by

\begin{align*}
W(n) = r\cdot W(n/q) + O(n^2).
\end{align*}

Solving the recurrence gives $W(n) = O\left(n^{\log_qr}\right)$ work for the entire algorithm. The space usage is also $O\left(n^{\log_qr}\right)$.
\end{proof}

Using Lemma~\ref{lem:matrix-mult-tensor-parallel}, we obtain the following parallel matrix multiplication bounds:

\begin{corollary}\label{cor:matrix-exp}
There exists a parallel matrix multiplication algorithm based on~\cite{Williams12,LeGall14} that multiplies two $n \times n$ matrices
with $O\left(n^{2.373}\right)$ work and $O(\log n)$ depth, using $O\left(n^{2.373}\right)$ space.
\end{corollary}

\fi
\section{Experimental Results}\label{sec:exps}

\myparagraph{Experimental Setup} Our experiments are performed on a 72-core
Dell PowerEdge R930 (with two-way hyper-threading) with $4\times 2.4\mbox{GHz}$
Intel 18-core E7-8867 v4 Xeon processors (with a 4800MHz bus and 45MB L3 cache)
and 1\mbox{TB} of main memory. Our programs use a work-stealing
scheduler that we implemented~\cite{BlAnDh20}. The scheduler is implemented
similarly to Cilk for parallelism. Our programs are compiled using
\texttt{g++} (version 7.3.0) with the \texttt{-O3} flag.

\myparagraph{Graph Data}
Table~\ref{table:sizes} lists the graphs that we use.
\defn{com-Orkut} is an undirected graph of the Orkut social
network~\cite{leskovec2014snap}.
\defn{Twitter} is a directed graph of the Twitter network~\cite{kwak2010twitter}.  We
symmetrize the Twitter graph for our experiments.
For some of our experiments which ingest a stream of edge updates, we
sample edges from an rMAT generator~\cite{chakrabarti2004r} with
$a=0.5, b=c=0.1, d=0.1$ to perform the updates. The
update stream can have duplicate edges, and
Table~\ref{table:rmatduplicate} reports the number of unique edges
found in prefixes of various sizes of the rMAT stream that we
generate. The unique edges in the full stream represents the rMAT
graph described in Table~\ref{table:sizes}.

\begin{table}[!t]
\centering
\footnotesize
\begin{tabular}[!t]{l|r|r}
\toprule
{Graph Dataset} & Num. Vertices & Num. Edges \\
\midrule
{Orkut        }    & 3,072,627        &234,370,166   \\
{Twitter        }    & 41,652,231       &2,405,026,092 \\
{rMAT}                    & 16,384           &121,362,232 \\
\end{tabular}
\caption{Graph inputs, including number of vertices and edges.}
\label{table:sizes}
\end{table}

\begin{table}[!t]\footnotesize
\footnotesize
\centering
\begin{tabular}{c|c||c|c}
  \toprule
$m$ & unique edges & $m$ & unique edges \\ 
\midrule
$2\times 10^{6}$ & 1,569,454  & $4\times 10^{8}$  & 55,395,676\\
 $2\times 10^{7}$ & 9,689,644 &    $8\times 10^{8}$ &  74,698,492 \\
  $1\times 10^{8}$  & 27,089,362 &    $3.2\times 10^{9}$ & 121,362,232\\
   $2\times 10^{8}$ & 39,510,764\\
   
\end{tabular}
\caption{Number of unique edges in the first $m$ edges from the \emph{rMAT} generator. }
\label{table:rmatduplicate}
\end{table}

\ifFull
\subsection{Our Implementation}\label{sec:ourimpl}
\fi

\myparagraph{Parallel Primitives}
We implemented a multicore CPU version of our algorithm using the
Graph Based Benchmark Suite (GBBS)~\cite{dhulipala2018theoretically},
which includes a number of useful parallel primitives, including
high-performance parallel sorting, and primitives such as
prefix sum, reduce, and
filter~\cite{JaJa92}. In what follows, a \defn{filter} takes an array $A$
and a predicate function $f$, and returns a new array containing $a
\in A$ for which $f(a)$ is true, in the same order that they appear in
$A$.  Our implementations use the atomic compare-and-swap and
atomic-add instructions available on modern CPUs.

\myparagraph{Implementation}
For $\cT$, we used the concurrent linear probing hash table by Shun
and Blelloch~\cite{shun2014phase}.  For each of the data structures
$\hh$, $\hl$, $\lh$, and $\lowl$, we created an array of size $n$,
storing (possibly null) pointers to hash tables~\cite{shun2014phase}.
For an edge $(u,v)$ in one of the data structures, the value $v$ will
be stored in the hash table pointed to by the $u$'th slot in the
array.  We also tried using hash tables for both levels, but found it
to be slower in practice.
 For deletions, we used the
folklore \emph{tombstone} method. In this method, when an element is
deleted, we mark the slot in the table as a tombstone, which is a
special value. When inserting, we can insert into a tombstone, but we
have to first check until seeing an empty slot to make sure that we are not
inserting a duplicate key.  In the preprocessing phase of the
algorithm, instead of using approximate compaction, we used
filter. To find the last update for duplicate updates, we use
a parallel sample sort~\cite{ShunBlellochFinemanEtAl2012} to sort the
edges first by both endpoints, and then by timestamp. Then we use filter to remove duplicate updates.  When we initialize
the dynamic data structures, a vertex is considered high-degree if it
has degree greater than $2t_1$ and low-degree otherwise.

During minor rebalancing, a vertex only changes its status if its
degree drops below $t_1$ or increases above $t_2$ due to the batch update.
In major rebalancing, we merge our dynamic data structure and the
updated edges into a compressed sparse row (CSR) format graph and use
the static parallel triangle counting algorithm by  Shun and
Tangwongsan~\cite{ShunT2015} to
recompute the triangle count.  We then build a new dynamic data
structure from the CSR graph. We also implement several natural
optimizations which improve performance. To reduce the overhead of
using hash tables, we use an array to store the neighbors of vertices
with degree less than a certain threshold (we used $128$ in our
experiments).  Moreover, we only keep a single entry for
$(u,v)$ and $(v,u)$ in the wedges table $\cT$.

{
\setlength{\tabcolsep}{2pt}
\begin{table}[!t]
\footnotesize
\centering
\begin{tabular}[!t]{p{0.2\columnwidth}|l|c|c|c|c|c}
\toprule
 & & \multicolumn{5}{c}{\textbf{Batch Size}} \\
\textbf{Algorithm} & \textbf{Graph} & \textbf{$2\times 10^{3}$} & \textbf{$2\times 10^{4}$} & \textbf{$2\times 10^{5}$} & \textbf{$2 \times 10^{6}$} & $m$ \\
\specialrule{.2em}{.1em}{.1em}
\multirow{3}{0.2\columnwidth}{Ours (INS)} & Orkut   & 1.90e-3 & 4.76e-3 & 0.0235 & 0.168 & -- \\
& Twitter & \textbf{2.11e-3} & \textbf{7.10e-3} & \textbf{0.0430} & \textbf{0.366} & -- \\
& rMAT    & \textbf{6.42e-4} & \textbf{2.09e-3} & \textbf{8.62e-3} & 0.0618 & -- \\
\hline
\multirow{3}{0.2\columnwidth}{Makkar et al. (INS) \cite{Makkar2017}} & Orkut  & \textbf{9.76e-4} & \textbf{2.69e-3} & \textbf{0.0143} & \textbf{0.0830} & -- \\
& Twitter & time-out & 0.0644&0.437 &3.88 & -- \\
& rMAT    & 1.98e-3&6.90e-3 &0.012 & \textbf{0.0335} & -- \\
\specialrule{.2em}{.1em}{.1em}
\multirow{3}{0.2\columnwidth}{Ours (DEL)} & Orkut  & 1.80e-3&4.37e-3 &0.0189 & 0.124  & -- \\
& Twitter & \textbf{2.14e-3} & \textbf{7.76e-3} & \textbf{0.0486} & \textbf{0.385} & -- \\
& rMAT    & 6.48e-4&2.23e-3 &9.21e-3 &0.0723 & -- \\
\hline
\multirow{3}{0.2\columnwidth}{Makkar et al. (DEL) \cite{Makkar2017}} & Orkut  & \textbf{4.63e-4} & \textbf{1.46e-3} & \textbf{8.12e-3} & \textbf{0.0499} & -- \\
& Twitter & time-out & 0.0597&0.401 & 3.64& -- \\
& rMAT    & \textbf{4.47e-4} & \textbf{1.81e-3} & \textbf{5.12e-3} & \textbf{0.027} & -- \\
\specialrule{.2em}{.1em}{.1em}
\multirow{3}{*}{Static~\cite{ShunT2015}} & Orkut  & -- & -- & -- & -- & 1.027 \\
& Twitter & --  & -- & -- & -- & 32.1 \\
& rMAT    & --  & -- & -- & -- & 14.7 

\end{tabular}
\caption{
Running times (seconds) for our parallel \batchdynamic{} triangle counting
algorithm and Makkar et al.~\cite{Makkar2017}'s algorithm on 72 cores with hyper-threading.
We apply the edges in each graph as batches of edge
insertions (INS) or deletions (DEL) of varying sizes, ranging from $2 \times 10^{3}$ to $2
\times 10^{6}$, and report the average time for each batch size. The update time of Makkar et al. algorithm for Twitter batch size $2 \times 10^{3}$ is missing because the expriment timed out.
We also report the update time for the state-of-the-art static
triangle counting algorithm of Shun and Tangwongsan~\cite{ShunT2015},
which processes a single batch of size $m$. Note that for the Twitter
and Orkut datasets, all of the edges are unique.
However, for the rMAT dataset, batches can have duplicate
edges. For each batch size of each dataset, we list the fastest time in bold.   }
\label{table:ourtimes}
\end{table}
}

\myparagraph{Experiments}
Table~\ref{table:ourtimes} report the parallel running times on varying insertion and deletion batch sizes for our
implementation of our new parallel \batchdynamic{} triangle counting
algorithm designed. For the two graphs based on static
graph inputs (Orkut and Twitter), we generate updates for the
algorithm by representing the edges of the graph as an array, and
randomly permuting them. 
The algorithm is then run using batches of
the specified size. 
For insertions, we start with an empty graph and apply batches from the beginning to the end of the permuted array. For deletions, we start with the full graph and apply batches from the end to the beginning of the permuted array.
The table also reports the running time for the
GBBS implementation of the state-of-the-art static triangle counting
algorithm of Shun and Tangwongsan~\cite{ShunT2015, dhulipala2018theoretically}.

Across varying batch sizes, our algorithm achieves throughputs between 1.05--16.2 million
edges per second for the Orkut graph, 0.935--5.46 million edges per
second for the Twitter graph, and 3.08--32.4 million edges per second
for the rMAT graph. We obtain much higher throughput for the rMAT
graph due to the large number of duplicate edges found in this graph
stream, as illustrated in Table~\ref{table:rmatduplicate}.
We observe that in all cases, the average  time for processing a batch is smaller than
the running time of the static algorithm. The maximum speedup of our
algorithm over the static algorithm is $22709\times$ for the rMAT
graph with a deletion batch of size $2 \times 10^{3}$, but in general our
algorithm achieves good speedups across the entire range of batches
that we evaluate.

Lastly, Figure~\ref{fig:speed-up} shows the parallel speedup of our
algorithm with varying thread-count on the Orkut and Twitter graph, for a fixed
batch size of $2 \times 10^{6}$. Our algorithm achieves a maximum of
$74.73 \times$ speedup using 72 cores with hyper-threading for this experiment.

\begin{figure}[!t]
  \centering
      \includegraphics[width=\columnwidth]{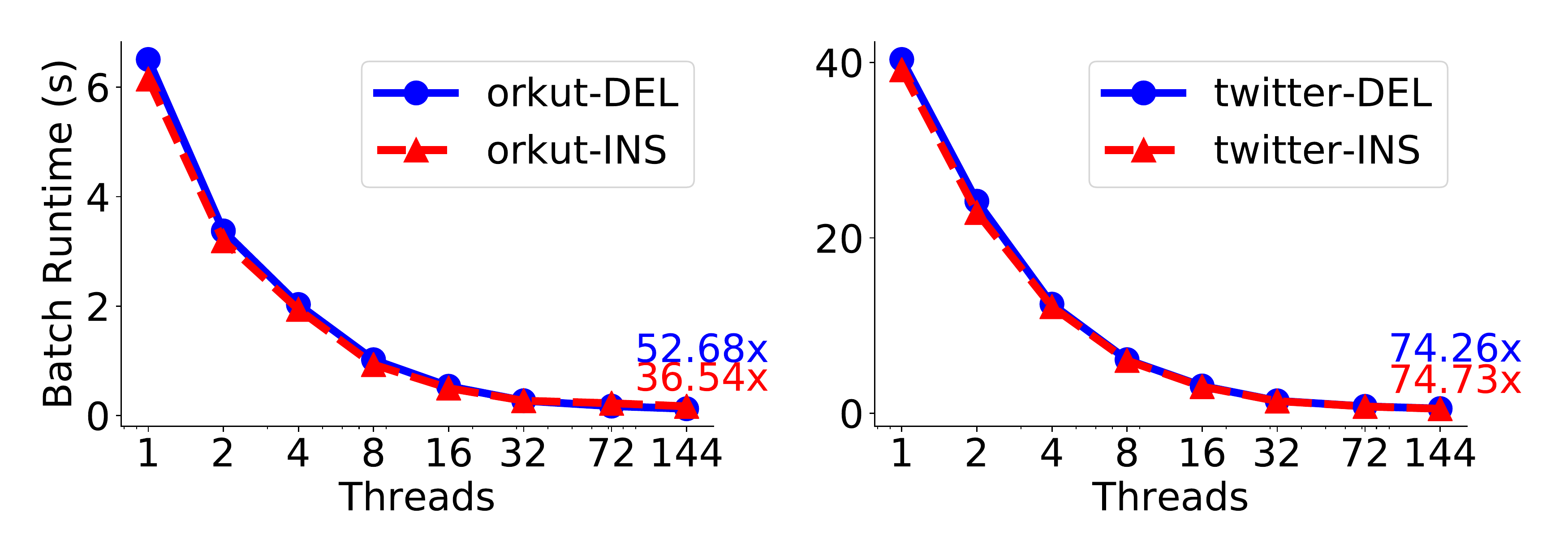}
  \caption{Running times of our parallel batch-dynamic triangle counting algorithm
  with respect to thread count (the $x$-axis is in log-scale) on the Orkut (average time across all batches) and Twitter (running time for the 6th batch) graph
  for both insertion (red dashed line)  and deletion (blue solid line).
  ``144'' indicates 72 cores with hyper-threading. The experiment is
  run with a batch size of $2\times 10^{6}$. The parallel speedup on 144 threads over a single thread is displayed.
  }\label{fig:speed-up}
\end{figure}

\ifFull
\subsection{Comparison with Existing Algorithms}\label{sec:makkar}
\fi

\myparagraph{Comparison with Ediger et al} We compared our
implementation with a shared-memory implementation of the Ediger et
al.\ algorithm~\cite{Ediger2010}, which is implemented as part of the
STINGER dynamic graph processing system~\cite{ediger2012stinger}.
Unfortunately, we found that their implementation is much slower than
ours due to bottlenecks in the update time for the underlying dynamic
graph data structure. We note that recent work on streaming graph
processing observed similar results for using
STINGER~\cite{dhulipala2019low}. To obtain a fair comparison, we chose
to focus on implementing a more recent GPU \batchdynamic{}
triangle counting algorithm ourselves, which we discuss next.

\myparagraph{Comparison with Makkar et al}
The Makkar et al.\ algorithm~\cite{Makkar2017} is a state-of-the-art
parallel \batchdynamic{} triangle counting implementation designed for
GPUs. To the best of our knowledge, there is no multicore
implementation of this algorithm, and so in this paper we implement an
optimized multicore version of their algorithm. The algorithm works as follows.
First, their algorithm separates the batch of updates into
batches for insertions and deletions. Then, for each batch of updates,
it creates an \emph{update graph}, $\hat{G}$, for each batch consisting
of only the updates within each batch. Then, it merges the updates
from each batch with the original edges in the graph to create an
updated graph for each of the batches, $G'$. Note that this graph
contains both the edges previously in the graph, as well as the new edges.

The merging process to construct $G'$ first sorts the batch to obtain
sorted lists of neighbors to add/delete from the adjacency lists of
vertices in the graph. Then, the algorithm performs a simple linear-work
procedure to merge each existing adjacency list with the sorted updates.
In particular, doing $t$ edge updates on a vertex with degree $d$ takes $O(d+t)$ work.
Finally, the algorithm counts the triangles by
intersecting the adjacency lists of the endpoints of each edge in the
batch.
For each edge $(u, v)$, they intersect $G'(u)$ with $G'(v)$, $G'(u)$
with $\hat{G}(v)$, and $\hat{G}(u)$ with $\hat{G}(v)$. The count of the
number of triangles can be obtained from the number of intersections
obtained from each of these cases using a simple inclusion-exclusion
formula. They provide a further optimization by only intersecting
\emph{truncated} adjacency lists in some of the cases where a
truncated adjacency list is one where the list only contains
vertices with IDs less than the ID of the vertex that the adjacency list
belongs to. Their algorithm has a worst case work bound of $O(n^2)$.

\myparagraph{Implementation}
We developed a new multicore implementation of the Makkar et al.\
algorithm
using the same parallel primitives and framework described earlier for
the implementation of our algorithm.
We implemented several optimizations
that improved performance. First, we handle vertices with degree
lower than $16$ by storing their incident edges in a special array of
size $16n$, and only allocate memory for vertices with larger degree.
Second, we note that their algorithm does not specify how to handle
redundant insertions that are already present in the graph. We
remove these edge updates by modifying the merge algorithm that
constructs $G'$ from $G$. Specifically, during the merge, if we
identify that a given edge is already present in $G$, we mark it in
the sorted sequence of batch updates that we are merging in.  Removing
these marked updates to construct $\hat{G}$ without redundant updates
is done by using a parallel filter.

\myparagraph{Performance Comparison}
Table~\ref{table:ourtimes} shows the running times of the Makkar et
al.\ algorithm on batches of insertions and deletions of different
sizes. The data points for the Twitter graph are also plotted in
Figure~\ref{fig:twitter}. We observe that the Makkar et al.\ algorithm
is faster than our algorithm on the Orkut graph, especially for large
batches. On the other hand, for the Twitter graph, our algorithm is
consistently faster for both insertions and deletions across all batch
sizes. This is because there are no vertices with very high degree in the Orkut graph, and so the Makkar et al.\ algorithm does less work in merging adjacency lists with updates, while the Twitter graph has vertices with extremely high degree, which are costly to merge.
Both algorithms are significantly faster than simply
applying the  static triangle counting algorithm for the range of batch sizes
that we considered.

\begin{figure}[!t]
  \centering
      \includegraphics[width=\columnwidth]{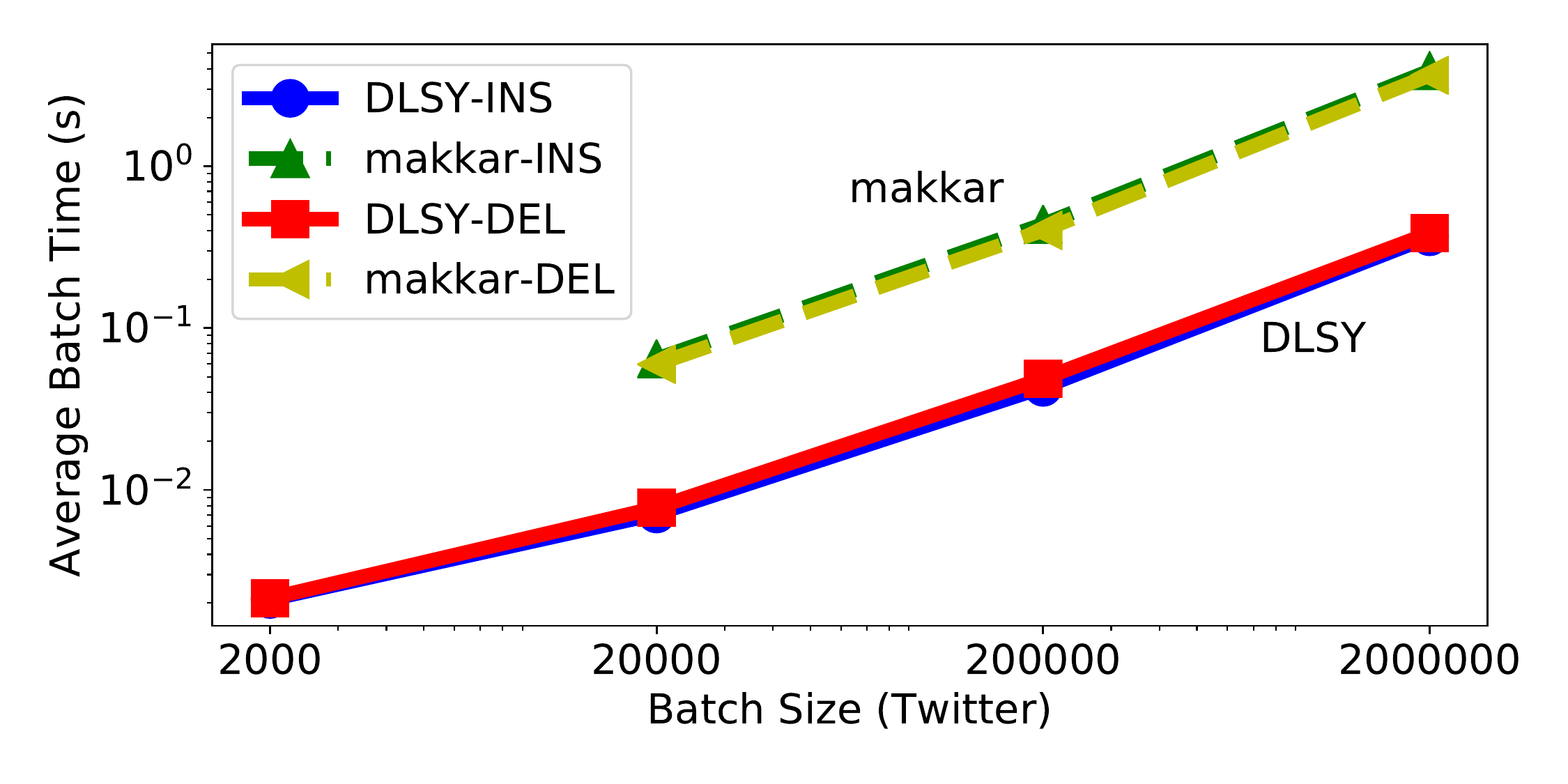}
  \caption{
This figure plots the average insertion and deletion
round times for each batch size (log-log scale) on Twitter using 72 cores with hyper-threading. The plot is in log-log scale. The lines for our algorithm are solid
(blue for insertion and red for deletion) while the lines for Makkar et al.
algorithm are dashed (green for insertion and yellow for
deletion). The update time of Makkar et al. algorithm for Twitter batch size $2
\times 10^{3}$ is missing because the experiment timed out (due to cumulative
runtime being too large).
  }\label{fig:twitter}
\end{figure}

Next, we evaluate the performance of insertion batches in our
algorithm and the Makkar et al.~algorithm on the synthetic rMAT graph
with 3.2 billion generated edges (which have duplicates). This
synthetic experiment allows us to study how both algorithms perform as
the graph becomes more dense. We evaluate the performance for
different insertion batch sizes. The experiment uses prefixes of the
rMAT graph (the number of unique edges per prefix is shown in
Table~\ref{table:rmatduplicate}) to control the density of the graph.
The vertex set in this experiment is fixed, and thus a
larger number of unique edges corresponds to a denser graph.

Figure~\ref{fig:makkar} plots the running time of both
implementations for varying batch sizes as a function
of the graph density. We observe that for small batch sizes, the
performance of the Makkar et al.\ algorithm degrades significantly as
the graph grows more dense and contains more high-degree vertices.
On
the other hand, our algorithm's performance generally does not degrade
as the graph grows denser, across all batch sizes.
We also significantly
outperform the Makkar et al.\ algorithm for small batch sizes.
Specifically, we obtain a maximum speedup of $3.31 \times$ for a batch
of size $2 \times 10^{4}$.
This is because the overhead
of updating of high-degree vertices in the Makkar et al.\ algorithm becomes
relatively higher, as work proportional to the vertex degree must be done regardless of the number of new incident edges.

\begin{figure}
  \centering
      \includegraphics[width=\columnwidth]{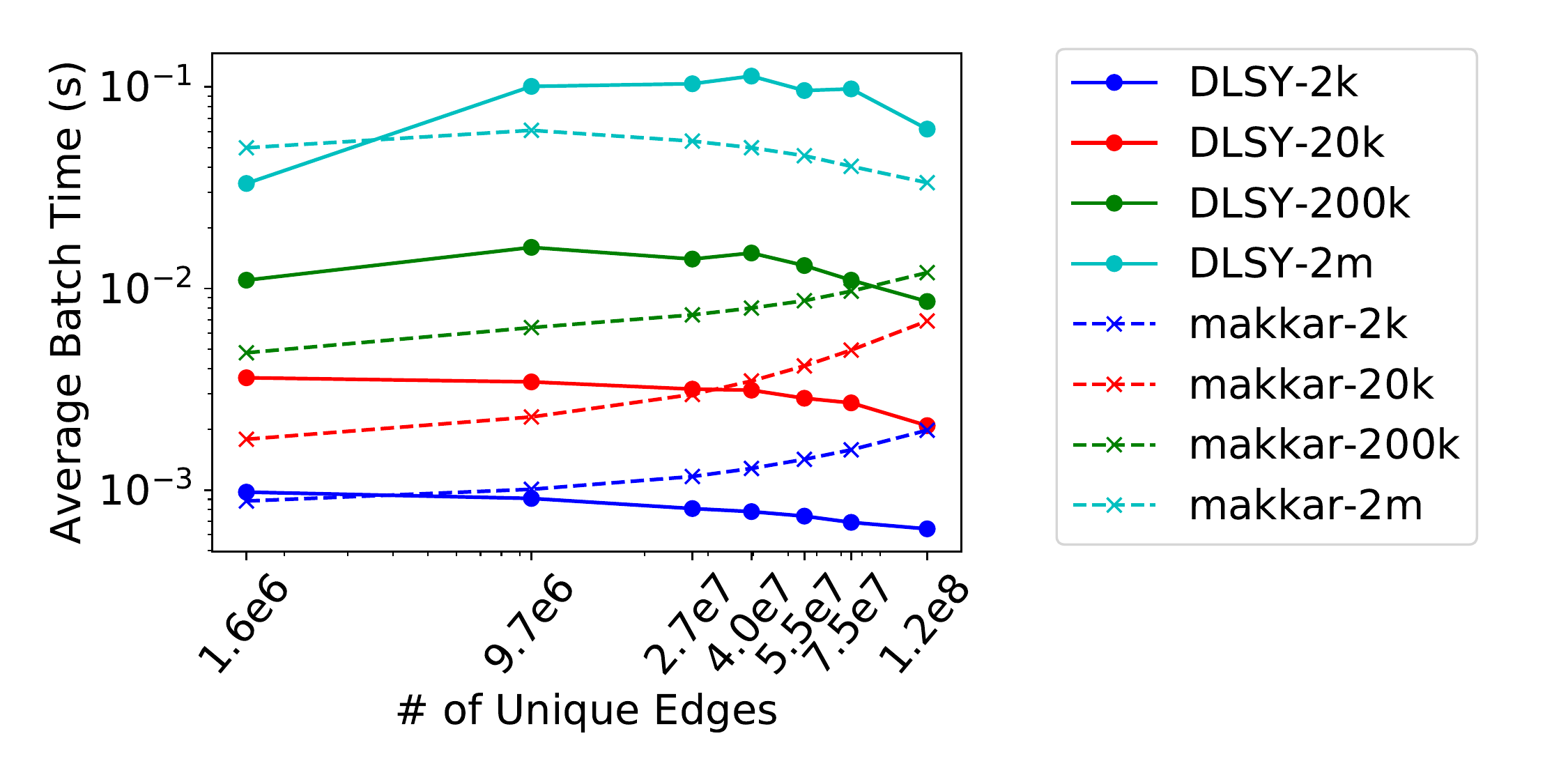}
  \caption{Comparison of the performance of our implementation (DLSY, solid line) and  Makkar et al.\ algorithm~\cite{Makkar2017} (makkar, dotted line) for batches of insertions.
The figure shows the average batch time for different batch sizes on the rMAT graph with varying prefixes of the generated edge stream to control density. The number of unique edges in the prefix is shown on the $x$-axis. The number of vertices is fixed at 16,384.
The dark blue, red, green, and light blue lines are for batches of size  $2\times 10^{3}$,  $2\times 10^{4}$,  $2\times 10^{5}$,
and  $2\times 10^{6}$, respectively.
  We see that our new algorithm is faster for small batches and on denser graphs.
  }\label{fig:makkar}
\end{figure}

\section{Conclusion}

In this paper, we have given new dynamic algorithms for the $k$-clique
problem.  We study this fundamental problem in the 
\batchdynamic{} setting, which is better suited for parallel hardware
that is widely available today, and enables dynamic algorithms to
scale to high-rate data streams.  We have presented a work-efficient
parallel \batchdynamic{} triangle counting algorithm.
We also gave a simple, enumeration-based algorithm for maintaining the
$k$-clique count.  In addition, we have presented a novel parallel
\batchdynamic{} $k$-clique counting algorithm based on fast matrix
multiplication, which is asymptotically faster than existing dynamic
approaches on dense graphs.
Finally, we provide a multicore implementation of
 our parallel
 \batchdynamic{} triangle counting algorithm and compare it with state-of-the-art implementations that have weaker theoretical guarantees, showing that
 our  algorithm is competitive in practice.

\ifFull
\fi

\myparagraph{Acknowledgements}
We thank Josh Alman, Nicole Wein, and Virginia Vassilevska Williams
for helpful discussions on various aspects of our paper.  We also
thank anonymous reviewers for their helpful suggestions.  This
research was supported by DOE Early Career Award \#DE-SC0018947, NSF
CAREER Award \#CCF-1845763, Google Faculty Research Award, DARPA SDH
Award \#HR0011-18-3-0007, and Applications Driving Architectures (ADA)
Research Center, a JUMP Center co-sponsored by SRC and DARPA.

\bibliographystyle{alpha}
\bibliography{ref}

\ifFull
\appendix
\section{Sequential Fully Dynamic Triangle Counting of~\cite{KNNOZ19}}\label{app:triangle}

Here, we present the sequential fully dynamic triangle counting
algorithm of Kara et al.~\cite{KNNOZ19} that operates in $O(m)$ space,
$O(\sqrt{m})$ amortized work per edge update, and $O(m^{3/2})$ work
for preprocessing. This algorithm returns the exact count of the
number of triangles in an undirected graph under both edge insertions
and deletions.
Kara et al.~\cite{KNNOZ19}
present their algorithm for directed $3$-cycles using relational database terminology (where each edge in the triangle may be drawn from a different relation),
but we simplify their algorithm for the case of undirected
 graphs. Kara et al.~\cite{KNNOZ19} prove the following theorem.

\begin{theorem}[Fully Dynamic Triangle Counting~\cite{KNNOZ19}]\label{thm:knnoz19}
There exists a sequential algorithm to count the number of triangles
in an undirected graph $G = (V, E)$ using $O(m^{3/2})$ preprocessing
work that can handle an edge update in $O(\sqrt{m})$
amortized work and $O(m)$ space.
\end{theorem}

We now explain the fully dynamic triangle counting algorithm
of~\cite{KNNOZ19} in greater detail.

Given a graph $G= (V, E)$ with $n = |V|$ vertices and $m = |E|$ edges,
we initialize the following variables: $M = 2m+1$, $t_1 = \sqrt{M}/2$,
and $t_2 = 3\sqrt{M}/2$.  We define a vertex to be \defn{low-degree}
if its degree is at most $t_1$ and \defn{high-degree} if its degree
is at least $t_2$. Vertices with degree in between $t_1$ and $t_2$ can
be classified either way.  Let $C$ be the current count of the number
of triangles in the graph. We compute the initial count of the number
of triangles in the input graph $G$ using a static triangle counting
algorithm~\cite{IR77} in $O(m^{3/2})$ work and $O(m)$ space. Thus, we
immediately have a preprocessing work of $O(m^{3/2})$.

We create four data structures $\hh$, $\hl$, $\lh$, and $\lowl$.
$\hh$ stores all of the edges $(u, v)$ where both $u$ and $v$ are
high-degree, $\hl$ stores edges $(u,v)$, where $u$ is high-degree and
$v$ is low-degree, $\lh$ stores the edges $(u, v)$ where $u$ is
low-degree and $v$ is high-degree, and $\lowl$ stores edges where both
$u$ and $v$ are low-degree. With our data structures, the
following operations are supported:

\begin{enumerate}
\item Given a vertex $v$, determine whether it is low-degree or high-degree in $O(1)$ work.
\item Given an edge $(u, v)$, check if it is in $\hh$, $\hl$, $\lh$, or $\lowl$ in $O(1)$ work.
\item Given a vertex $v$, return all neighbors of $v$ in $\hh$, $\hl$, $\lh$, and $\lowl$ in $O(\deg(v))$ work.
\item Given an edge $(v, w)$ to insert or delete, update $\hh$, $\hl$, $\lh$, or $\lowl$ in $O(1)$ work.
\end{enumerate}

We can implement $\hh$, $\hl$, $\lh$, and $\lowl$ to support these
operations by using a two-level hash table for each of these
structures and an additional array $\degarray$. $\degarray$ is a
dynamic hash table containing a key for each vertex that has non-zero
degree and stores the degree of the vertex as the value. The data
structures support insertions and deletions in $O(1)$ work.
$\degarray$ can be initialized in $O(m)$ work by scanning over all
vertices and computing their degree.  $\hh$, $\hl$, $\lh$, and $\lowl$
can be initialized in $O(m)$ work by scanning over all edges and
inserting them into the right table based on the degrees of their
endpoints.

We maintain one additional data structure $\cT$ that counts the number
of wedges $(u, w, v)$, where $u$ and $v$ are high-degree vertices and
$w$ is a low-degree vertex. $\cT$ has the property that given an edge
insertion or deletion $(u, v)$ where both $u$ and $v$ are
high-degree vertices, it returns the number of such wedges $(u, w, v)$
where $w$ is low-degree that $u$ and $v$ are part of in $O(1)$
work. We can implement this via a hash table indexed by pairs of
high-degree vertices that stores the number of wedges for each
pair.
$\cT$ can be initialized in
$O(m^{3/2})$ work by iterating over all edges $(u, w)$ in $\hl$ and
then for each $w$, iterating over all edges $(w, v)$ in $\lh$ to
determine whether $v$ is high-degree, and if so then increment
$T(u, v)$ by $1$.  There are $O(m)$ edges $(u,w)$ in $\hl$, and for
each $w$ there are at most $O(\sqrt{m})$ edges $(w,v)$ in $\lh$ since $w$ is
low-degree. Each lookup and increment takes $O(1)$ work, giving an
overall work of $O(m^{3/2})$.

\ifCameraReady
\subsection{Update Procedure~\cite{KNNOZ19}}\label{app:triangle-updates}~
\fi
\ifFull
\subsection{Update Procedure~\cite{KNNOZ19}}\label{app:triangle-updates}
\fi

The procedure for handling single edge updates in the sequential
setting given by~\cite{KNNOZ19} as follows:

For an edge insertion (resp.\ deletion) $(u, v)$, we first find the
degree of $u$ and $v$ in $\degarray$ and then look up the edge in
their respective tables $\hh$, $\hl$, $\lh$, or $\lowl$.  If the edge
already exists (resp.\ does not exist) in the table, nothing else is
done. Otherwise, we need to find all tuples $(u, w, v)$ such that $(v,
u)$ and $(u, w)$ already exist in the graph because for each such
tuple, a new triangle will be formed (resp.\ an existing triangle will be
deleted). We first update the triangle count, and then we update the data
structures. For updating the triangle count $C$, there are $4$ different
cases for such tuples, and so we check each of the following cases:

\begin{enumerate}
    \item \textbf{$(u, w)$ is in $\hh$ and $(w, v)$ is in $\high y$ where $y \in \{\high, \low\}$}: We extract all high-degree neighbors of $u$ in $\hh$. Given that the degree of all high-degree vertices is $\Omega(\sqrt{m})$, there are at most $O(\sqrt{m})$ such vertices. For each of these neighbors, we can check in $O(1)$ work for each $w$ whether $(w, v)$ exists in $\high y$. This takes $O(\sqrt{m})$ work.
    \item \textbf{$(u, w)$ is in $\hl$ and $(w,v)$ is in $\lh$ where $y \in \{\high, \low\}$}: Since both $u$ and $v$ are high-degree in this case, we perform an $O(1)$ work lookup in $\cT$ for the count of the number of wedges $(u, w, v)$ in this case.
    \item \textbf{$(u, w)$ is in $\lh$ and $(w, v)$ is in $\high y$ where $y \in \{\high, \low\}$}: Scan through the neighbors of $u$ in $\lh$. For each neighbors of $u$, check whether $(w, v)$ exists in $\high y$. This takes $O(\sqrt{m})$ work since $u$ has low-degree.
    \item \textbf{$(u, w)$ is in $\lowl$ and $(w, v)$ is in $\low y$ where $y \in \{\low, \high\}$}: Again, scan through the neighbors of $u$ in $\lh$. For each neighbors of $u$, check whether $(w, v)$ exists in $\low y$. This takes $O(\sqrt{m})$ work since $u$ has low-degree.
\end{enumerate}

After updating the triangle count, we proceed with updating the data
structures with the edge insertion (resp.\ deletion).

We first update $\cT$ given an edge insertion (resp.\ deletion) $(u, v)$ as follows:

\begin{enumerate}
\item If $u$ is high-degree and $v$ is low-degree, then we find all of $v$'s neighbors in $\lh$ and for each such neighbor $x$, we increment (resp.\ decrement) the entry $\cT(u, x)$ by $1$. It takes $O(\sqrt{m})$ work to perform this update since $v$ is low-degree.
\item If $u$ is low-degree and $v$ is high-degree, then we scan through all vertices in $\hl$ and for each vertex $x$ in $\hl$ that has $u$ as a neighbor, we increment (resp.\ decrement) $\cT(x, v)$ by $1$. This takes $O(\sqrt{m})$ work since there are at most $O(\sqrt{m})$ high-degree vertices.
\end{enumerate}

In addition to the updates to $\cT$, we also insert (resp. delete) $(u, v)$ into $\hh$, $\hl$, $\lh$, and $\lowl$ depending on the degrees of $u$ and $v$, and update $\degarray$.
For a given edge $(u, v)$
insertion (resp.\ deletion), we first determine whether $u$ and $v$
are low-degree or high-degree by looking in $\degarray$ for $u$ and
$v$ in $O(1)$ work. $\hh$, $\hl$, $\lh$, and $\lowl$ are constructed
as hash tables keyed by first the first vertex in the edge tuple and
then the second vertex in the edge tuple with pointers to second-level
hash tables storing the neighbors of that particular vertex. If $u$ is
high-degree, then the edge is inserted (resp. deleted) into $\hh$ or
$\hl$ (depending on whether $v$ is low or high-degree) using $u$ as
the key and adding $v$ to the second level hash table. Similarly, if
$u$ is low-degree, $(u, v)$ is inserted (resp. deleted) into $\lh$ or
$\lowl$. Furthermore, $(v, u)$ is also inserted into its respective
table depending on whether $v$ is low or high-degree. The entries for
$u$ and $v$ in $\degarray$ are then incremented (resp. decremented) in
$\degarray$. The updates to these data structures take $O(1)$ work.

We also have to deal with the cases where the degree classification of
vertices have changed or the number of edges has changed by too much
that the values of $M$, $t_1$, and $t_2$ need to be updated.  This is
described in the next section.

\ifCameraReady
\subsection{Rebalancing~\cite{KNNOZ19}}\label{sec:rebalancing}~
\fi
\ifFull
\subsection{Rebalancing~\cite{KNNOZ19}}\label{sec:rebalancing}
\fi

We now describe the rebalancing procedure given in~\cite{KNNOZ19} when
a low-degree vertex becomes a high-degree vertex (or vice versa) and
when too many updates have been applied (and all the data structures
must be changed according to the new values of $M$, $t_1$, and $t_2$).

\myparagraph{Minor rebalancing} This type of rebalancing occurs if a vertex which was previously high-degree has its degree fall below $t_1$ or if a vertex that was previously low-degree has its degree increase above $t_2$. In the first case, we move the vertex and all its edges from $\hh$ to $\hl$, and from $\lh$ to $\lowl$. In the second case, we move the vertex and all its edges from $\hl$ to $\hh$, and from $\lowl$ to $\lh$. Since our data structures support additions and deletions of an edge in $O(1)$ work, and since the degree of $v$ is $\Theta(\sqrt{m})$ at this point, we perform $\Theta(\sqrt{m})$ updates. We showed in Section~\ref{app:triangle-updates} that updates take $O(\sqrt{m})$ work so we take $O(m)$ work overall for a minor rebalancing. However, $\Omega(\sqrt{m})$ updates must have occurred on this vertex before we have to perform minor rebalancing since $t_2-t_1 = \Theta(\sqrt{m})$,
and so we can amortize this cost over the $\Omega(\sqrt{m})$ updates,
resulting in $O(\sqrt{m})$ amortized work per update.

\myparagraph{Major rebalancing} A major rebalancing occurs when $m$, the number of edges in the graph, falls outside the range $[M/4, M]$. We simply reinitialize the data structures as in the original algorithm. Major rebalancing can only occur after $\Omega(M)$ updates, and so we can afford to re-initialize our data structure and recompute the triangle count from scratch using an $O(m^{3/2})$ work triangle counting algorithm. The amortized work of major rebalancing over $\Omega(m)$ updates is then $O(\sqrt{m})$.

\fi

\end{document}